\documentclass[a4paper,reqno]{amsart}
\usepackage{amssymb,latexsym}
\usepackage{amsfonts}
\usepackage{amsmath}
\usepackage[colorlinks,linkcolor=blue,anchorcolor=blue,citecolor=blue,urlcolor=blue]{hyperref}
\usepackage{algorithm}
\usepackage{enumerate}
\usepackage{algpseudocode}
\usepackage{verbatim}
\usepackage{graphicx}
\usepackage{amsthm}

\usepackage{lineno}

\def\F{\mathbb{F}}
\def\E{\mathbb{F}_{q_1}}

\def\Tr{\text{\rm Tr}}

\def\Ker{\text{\rm Ker}}

\def\a{\textrm{a}}
\newtheorem{theorem}{Theorem}[section]
\newtheorem{lemma}[theorem]{Lemma}
\newtheorem{example}[theorem]{Example}
\newtheorem{corollary}[theorem]{Corollary}
\newtheorem{proposition}[theorem]{Proposition}
\newtheorem{definition}[theorem]{Definition}

\numberwithin{equation}{section}

{\centering
\title[Complete weight enumerators and weight hierarchies]{Complete weight enumerators and weight hierarchies for linear codes from quadratic forms }
}
\author{Xiumei Li}
\address{School of Mathematical Sciences, Qufu Normal University, Qufu, 273165, China}
\email{lxiumei2013@qfnu.edu.cn}

\author{Xiaotong Sun}
    \address{School of Mathematical Sciences, Qufu Normal University, Qufu  273165, China.}\email{1905654165@qq.com}
    \author{Min Sha}
\address{\parbox{\linewidth}{School of Mathematical Sciences, South China Normal University, Guangzhou, 510631, China \\ 
Key Laboratory of Applied Mathematics (Putian University), Fujian Province University, Fujian Putian, 351100, China}}
\email{min.sha@m.scnu.edu.cn}



\date{\today}
\begin{document}
\maketitle
\begin{abstract}
In this paper, for an odd prime power $q$, we extend the construction of Xie et al. \cite{XOYM2023} to propose two classes of linear codes $\mathcal{C}_{Q}$ and $\mathcal{C}_{Q}'$ over the finite field $\mathbb{F}_{q}$ with at most four nonzero weights. These codes are derived from quadratic forms through a bivariate construction. We completely determine their complete weight enumerators and weight hierarchies by employing exponential sums. Most of these codes are minimal and some are optimal in the sense that they meet the Griesmer bound. Furthermore, we also establish the weight hierarchies of $\mathcal{C}_{Q,N}$ and $\mathcal{C}_{Q,N}'$, which are the descended codes of $\mathcal{C}_{Q}$ and $\mathcal{C}_{Q}'$.
\end{abstract}

\section{Introduction}

\label{intro}


Throughout this paper, let $ \mathbb{F}_{q^s} $ be the finite field with size $ q^{s}$, where $q = p^m$ is an odd prime power and $s, m$ are two positive integers. Denote by $\mathbb{F}_{q^s}^{*}$ the set of the nonzero elements of $ \mathbb{F}_{q^s}$ and by $\Tr_{q^{s}/q}$ the trace function from $\mathbb{F}_{q^s}$ onto $\mathbb{F}_q$. A $k$-dimensional subspace $\mathcal{C}$ of $ \mathbb{F}_{q}^{n} $ over $\F_q$ is called an $[n,k,d]_q$ linear code of length $n$ with  minimum (Hamming) distance $d$. For $i \in \{1,2,\cdots,n\}$, let $A_{i}$ be the number of codewords in $\mathcal{C}$ with Hamming weight $i$. The sequence $(1,A_{1},\cdots,A_{n})$ is referred to as weight distribution of $\mathcal{C}$.
If there are $t$ nonzero elements in the sequence $(A_{1},A_{2},\cdots,A_{n})$, then the code $\mathcal{C}$  is said to be a $t$-weight code.

The weight distribution and complete weight enumerator are important parameters for a linear code \cite{T2007}. In particular, few-weight linear codes have broad applications in areas such as authentication codes, association schemes, secret sharing schemes, strongly regular graphs, and combinatorial designs \cite{CG1984,CK1986,DHKW2007,O2018,YD2006}. Therefore, determining the complete weight enumerators of linear codes remains a central research topic in coding theory.

Motivated by applications to cryptography, the concept of generalized Hamming weight (GHW) was initially introduced by Helleseth et al. \cite{HKM77,K78}.  We recall the definition of the generalized Hamming weights of linear codes \cite{WJ91}. For an $[n,k,d]_q$ linear code $\mathcal{C}$, 
let
$ [\mathcal{C},r]_{q} $ be the set of all the $r$-dimensional subspaces of $\mathcal{C}$, where $ 1\leq r\leq k$.
For a subspace $ K \in [\mathcal{C},r]_{q}$, the support of $K$ is defined by
\[ \textrm{Supp}(K)=\Big\{i:1\leq i\leq n, c_i\neq 0 \ \ \textrm{for some $c=(c_{1}, c_{2}, \cdots , c_{n})\in K$}\Big\}.\]
The $r$-th generalized Hamming weight of $\mathcal{C}$ is defined by
\[
d_{r}(\mathcal{C})=\min \Big\{\#\textrm{Supp}(K)| K \in [\mathcal{C}, r]_{q} \Big\}, \ 1 \leq r \leq k.
\]
The sequence $ \{d_{1}(\mathcal{C}),d_{2}(\mathcal{C}),\cdots,d_{k}(\mathcal{C})\}$ is called the weight hierarchy of $\mathcal{C}$. 

The weight hierarchy of a linear code is not only theoretically significant but also find applications in wire-tap channels, trellis complexity analysis, and network coding \cite{CGHFRS85,F94,KTFL93,TV95,WJ91}. While the weight hierarchies of several classical code families (e.g., Reed-Muller codes, cyclic codes, and algebraic geometry codes) have been fully characterized \cite{B19,BLV14,CC97,HP98,JL97,WZ94,WW97,WJ91,XLG16,YL15}, determining the complete weight hierarchies for general linear codes remains highly challenging. To date, only a limited number of results are known for codes constructed via defining sets \cite{HLL24,JF17,LF18,LF21,LL21,LL22dcc,LL22dm,LZW23,LW19}.

Xie et al. \cite{XOYM2023} constructed two classes of $\mathbb{F}_{q}$-linear codes $\mathcal{C}_{Q}$ and $\mathcal{C}_{Q}'$ as follows: 
\begin{linenomath}\begin{align*}\mathcal{C}_{Q} = \{(a Q(x) - \Tr_{q^{s}/q}(b x))_{x \in \F_{q^{s}}^*}: (a, b) \in \F_q \times \F_{q^{s}}\},\\
 \mathcal{C}_{Q}' = \{(a Q(x) - \Tr_{q^{s}/q}(b x) + c)_{x \in \F_{q^{s}}}: (a, b, c) \in \F_q \times \F_{q^{s}} \times \F_q\}, \end{align*}\end{linenomath}
where $Q$ is an arbitrary quadratic form on $\mathbb{F}_{q^{s}}/\mathbb{F}_{q}$. They demonstrated that these codes are two-, three-, or four-weight codes and include optimal codes attaining both the Griesmer and Singleton bounds. Furthermore, they introduced a novel technique to descend an $\mathbb{F}_{q}$-code to an $\mathbb{F}_{p}$-code, whose weight enumerator is uniquely determined by the original code. For analogous constructions based on specific quadratic forms, we refer readers to \cite{CDY05,DLLZ16,FL07,LLQ09,YYZ17,YCD06}.

Inspired by \cite{TXF2017,XOYM2023}, this paper defines two new classes of $\mathbb{F}_{q}$-linear codes $\mathcal{C}_{Q}$ and $\mathcal{C}_{Q}'$ as follows:   
\begin{linenomath}\begin{align*}\mathcal{C}_{Q} = \{(a Q(x) + \Tr_{q^{m_{2}}/q}(b y))_{(x,y) \in \F^\star}: (a, b) \in \F_q \times \F_{q^{m_2}}\},\\
 \mathcal{C}_{Q}' = \{(a Q(x) + \Tr_{q^{m_{2}}/q}(b y) + c)_{(x,y) \in \F}: (a, b, c) \in \F_q \times \F_{q^{m_2}} \times \F_q\}, \end{align*}\end{linenomath}
where $Q$ is an arbitrary quadratic form on $\mathbb{F}_{q^{m_1}}/\mathbb{F}_{q}$ and $m_{1}, m_{2}$ are arbitrary positive integers and $\F^\star  = \left(\F_{q^{m_1}} \times \F_{q^{m_2}}\right)\Big\backslash \{(0, 0)\}$.
The main contribution of this paper is the complete determination of the weight enumerators and weight hierarchies of these codes through exponential sum. This approach fundamentally differs from Xie et al.'s methods.

The organization of this paper is as follows: Section 2 introduces the essential notations and reviews fundamental properties of quadratic forms over finite fields. Section 3 determines the parameters of the proposed linear codes $\mathcal{C}_{Q}$ and $\mathcal{C}_{Q}'$, explicitly characterizes their weight enumerators, and fully establishes their weight hierarchies. Section 4 investigates the descent operation from $\mathcal{C}_{Q}$ and $\mathcal{C}_{Q}'$ to $\mathcal{C}_{Q,N}$ and $\mathcal{C}_{Q,N}'$ over prime fields, and derives the weight hierarchies of the resulting codes. Section 5 concludes the paper with a summary of contributions.

\section{Preliminaries}

\subsection{Notations and conventions}

Throughout this paper, we use the following notations:
\begin{itemize}
    \item $p$ is an odd prime and $q = p^m$.
    \item $m_1, m_2$ are positive integers, and $M = m_1 + m_2$.
    \item $q_i = q^{m_i}$ for $i=1,2$, $\mathbb{F} = \mathbb{F}_{q_1} \times \mathbb{F}_{q_2}$, and $ \mathbb{F}^\star = \mathbb{F} \setminus \{(0,0)\}$.
    \item $\zeta_p = \exp\left(\frac{2\pi \mathrm{i}}{p}\right)$ is a primitive $p$-th root of unity.
    \item \ $p^* = (-1)^{(p-1)/2} p$.
    \item $\eta$ denotes the quadratic character of $\mathbb{F}_q$, with $\eta(0) = 0$.
    \item $\upsilon: \mathbb{F}_q \to \mathbb{Z}$ is the function defined by $\upsilon(x) = q - 1$ if $x = 0$, and $ \upsilon(x) = -1$ otherwise.
    \item $\langle \alpha_1, \dots, \alpha_r \rangle$ denotes the $\mathbb{F}_q$-linear subspace spanned by $ \alpha_1, \dots, \alpha_r$.
\end{itemize}

\begin{lemma}[{\cite[Lemma 2.2]{LCL24}}]\label{lem:basic}
    For any $b \in \mathbb{F}_q$,
    \[
    \sum\limits_{z \in \mathbb{F}_q^*} \eta(z)^k \zeta_p^{\Tr_{q/p}(zb)} = 
    \begin{cases}
        \upsilon(b), & \text{if } k \text{ is even}, \\
        (-1)^{m-1} \eta(-b) q (p^*)^{-\frac{m}{2}}, & \text{if } k \text{ is odd}.
    \end{cases}
    \]
\end{lemma}

\subsection{Quadratic forms over finite fields}

The general references for this subsection are \cite{ChenMes2023, HK2006, TXF2017, XOYM2023}.

\setlength{\parindent}{2em} A quadratic form $Q$ on $\E/\mathbb{F}_{q}$ is a map $\E \to \F_q$ such that $Q(ax) = a^2 Q(x)$ for all $a \in \F_q$ and \begin{linenomath}\begin{equation*}\label{eq:F}
    B_{Q}(x,y) = \frac{1}{2} \Big(Q(x + y) - Q(x) - Q(y)\Big), \textrm{for any}\  x,y\in\E
\end{equation*}\end{linenomath} is a symmetric bilinear form. 
Taking a basis $\upsilon_{1}, \upsilon_{2}, \cdots, \upsilon_{m_1}$ of $\E$, there is a symmetric matrix $ A = (B_{Q}(\upsilon_i,\upsilon_j))_{m_1\times m_1}$ over $\mathbb{F}_{q}$ representing $Q(x)$ such that 
\begin{linenomath}\begin{align*}
Q(x) = Q(\overline{x}) = Q(x_{1}, x_{2}, \cdots, x_{m_1}) & = \sum\limits_{1\leq i, j\leq m_1} B_{Q}(v_i,v_j) x_{i} x_{j}
= \overline{x} A \overline{x}^{\mathrm{T}},
\end{align*}\end{linenomath}
where $\overline{x} = (x_1,x_2,\cdots,x_{m_1})$ is the coordinate vector of $x$ under the basis $\upsilon_{1}, \upsilon_{2},\cdots,\upsilon_{m_1}$ of $\E$ and $\overline{x}^{\mathrm{T}}$ is the transposition of $\overline{x}$. The rank $r_{Q}$ of $Q$ is defined as the codimension of $\F_q$-vector space 
\[\E^{\perp_Q}=\{ x \in \E | B_{Q}(x,y) = 0, \textrm{for all } y \in \E\},\] 
that is $r_{Q} = m_1 - \operatorname{dim}(\E^{\perp_Q})$. If $r_{Q} = m_1$, then $Q$ is called non-degenerate, otherwise $Q$ is called degenerate. It's known that every quadratic form can be diagonalized, i.e., there exists an invertible matrix $M$ over $\F_q$ such that
\[
M A M^{\mathrm{T}} = \operatorname{diag}(\lambda_{1},\lambda_{2},\cdots,\lambda_{r_{Q}},0,\cdots,0)
\]
is a diagonal matrix, where $\lambda_{1},\lambda_{2},\cdots,\lambda_{r_{Q}} \in \F_q^*$.
Let $\Delta_{Q} = \lambda_{1}\lambda_{2}\cdots\lambda_{r_{Q}}$ if $r_{Q} \neq 0$ and $\Delta_{Q} = 1$, otherwise. 
The sign $\varepsilon_{Q}$ of $Q$ is defined as $\eta(\Delta_{Q})$. For simplicity, we define $\epsilon = \varepsilon_{Q} (- 1)^{\frac{(p - 1) m r_{Q}}{4}}$ if $r_{Q}$ is even; and $\epsilon = \varepsilon_{Q} (- 1)^{\frac{(p - 1) m (r_{Q} + 1)}{4}}$ if $r_{Q}$ is odd. 

\begin{lemma}[{\cite[Lemma 2.4]{LCL24}, \cite[Lemma 5]{TXF2017}}]\label{lem:main1} 
  Let $Q(x)$ be a quadratic form on $\E/\mathbb{F}_{q}$. Then for $z \in \F_q^*$, we have 
\[\sum\limits_{x\in \mathbb{F}_{q_1}} \zeta_p^{\mathrm{Tr}_{q/p}\left(z Q(x)\right)} = \begin{cases}
		\epsilon q^{m_1 - \frac{r_{Q}}{2}}, & 2 \mid r_{Q}, \\
	(-1)^{m-1} \eta(- z) \varepsilon_Q q^{m_1}(p^*)^{-\frac{m r_{Q}}{2}}, & 2 \nmid r_{Q}.\end{cases}\]
\end{lemma}

\begin{lemma}\label{lem:main2}
 Let $Q(x)$ be a quadratic form on $\E/\mathbb{F}_{q}$. For any  element $(a,b) \in \mathbb{F}_{q} \times \mathbb{F}_{q_2}$ and $\beta \in \F_q$,
put
\[N(a, b; \beta) = \#\Big\{(x,y) \in \F: a Q(x) + \Tr_{q_2/q}(by) = \beta \Big\}.\] 
We have the following. 
\begin{enumerate}
\item [(1)] When $a = b = 0$, we have $N(0, 0; \beta) = q^M$ if $\beta = 0$, and $0$ otherwise.
\item [(2)] When $b \neq 0$, we have $N(a, b; \beta)=q^{M-1}$.
\item [(3)] When $a \neq 0, b = 0$, we have 
\[
N(a, b; \beta)
=\left\{\begin{array}{ll}
    		q^{M-1} \Big(1 + \epsilon q^{-\frac{r_{Q}}{2}} \upsilon(\beta)\Big), & 2 | r_{Q}, \\
    	q^{M-1} \Big(1 + \epsilon q^{\frac{1 - r_{Q}}{2}} \eta(-a \beta)\Big),  & 2 \not | r_{Q}.
    	\end{array}
    	\right.
\]
\end{enumerate}
\end{lemma}
\begin{proof} 
By the orthogonal property of additive characters, we have
\begin{linenomath}\begin{align*}
&N(a, b; \beta) = \frac{1}{q} \sum\limits_{(x,y) \in \mathbb{F}} \sum\limits_{z \in \mathbb{F}_{q}} \zeta_{p}^{\Tr_{q/p}\left(z(a Q(x) + \Tr_{q_2/q}(by) - \beta)\right)}  \\ 
& = \frac{1}{q} \sum\limits_{(x,y) \in \mathbb{F}}\left(1 + \sum\limits_{z \in \mathbb{F}_{q}^*}\zeta_{p}^{\Tr_{q/p}\left(z (a Q(x) + \Tr_{q_2/q}(b y) - \beta)\right)}\right)  \\ 
&=q^{M - 1} + \frac{1}{q} \sum\limits_{(x,y) \in \mathbb{F}} \sum\limits_{z \in \mathbb{F}_{q}^*} \zeta_{p}^{\Tr_{q/p}\left(z (a Q(x) + \Tr_{q_2/q}(b y) - \beta)\right)} \\ 
& = q^{M - 1} + \frac{1}{q} \sum\limits_{z \in \mathbb{F}_{q}^*} \zeta_p^{\Tr_{q/p}(- z \beta)}\sum\limits_{x \in \mathbb{F}_{q_1}} \zeta_{p}^{\Tr_{q/p}\left(z a Q(x)\right)}\sum\limits_{y\in \mathbb{F}_{q_2}} \zeta_{p}^{\Tr_{q/p}\left(\Tr_{q_2/q}(z b y)\right)} \\ 
& = q^{M - 1} + \frac{1}{q} \sum\limits_{z \in \mathbb{F}_{q}^*} \zeta_p^{\Tr_{q/p}(- z \beta)} \sum\limits_{x \in \mathbb{F}_{q_1}} \zeta_{p}^{\Tr_{q/p}\left(z a Q(x)\right)}\sum\limits_{y \in \mathbb{F}_{q_2}} \zeta_{p}^{\Tr_{q_2/p}(z b y)}.
\end{align*}\end{linenomath}

(1) When $a = b = 0$, we have $\beta = 0$ and $N(0, 0; 0) = q^M$.

(2) When $b \neq 0$, then $\sum\limits_{y \in \mathbb{F}_{q_2}} \zeta_{p}^{\Tr_{q_2/p}(z b y)} = 0$ and the desired conclusion is obtained.

(3) When $a \neq 0, b = 0$, we have 
\begin{linenomath}\begin{align*}
N(a, b; \beta) = q^{M - 1} + q^{m_2 - 1} \sum\limits_{z \in \mathbb{F}_{q}^*} \zeta_p^{\Tr_{q/p}(- z \beta)} \sum\limits_{x \in \mathbb{F}_{q_1}} \zeta_{p}^{\Tr_{q/p}\left(z a Q(x)\right)}.
\end{align*}\end{linenomath}
By Lemma~\ref{lem:main1}, we have
\begin{linenomath}\begin{align*}
\sum\limits_{x \in \mathbb{F}_{q_1}} \zeta_{p}^{\Tr_{q/p}(z a Q(x))} 
 = \left\{\begin{array}{ll}
    		\epsilon q^{m_1 - \frac{r_{Q}}{2}}, & 2|r_{Q}, \\
    	(-1)^{m-1} \varepsilon_Q \eta(-za) q^{m_1}(p^*)^{-\frac{m r_{Q}}{2}} ,  & 2\not|r_{Q}.
    	\end{array}
    	\right.
\end{align*}\end{linenomath}
It then follows from Lemma~\ref{lem:basic} that the desired conclusion is obtained.
    \end{proof}

From Lemma~\ref{lem:main2}, we can directely get the following corollary.

\begin{corollary}\label{pro:main3}
 Let $Q(x)$ be a quadratic form on $\E/\mathbb{F}_{q}$. For any  element $(a, b, c) \in \mathbb{F}_{q} \times \mathbb{F}_{q_2} \times \mathbb{F}_{q}$ and $\beta \in \F_q$,
put
\[N(a, b, c; \beta) = \#\Big\{(x,y) \in \F: a Q(x) + \Tr_{q_2/q}(b y) + c = \beta \Big\}.\] We have the following.
\begin{enumerate}
\item [\rm{(1)}] When $a = b = 0$, we have $N(0, 0, c; \beta) = N(0, 0; 0) = q^M$ if $\beta = c$, and $0$ otherwise.
\item [\rm{(2)}] When $b \neq 0$, we have $N(a, b, c; \beta) = N(a, b; \beta - c) = q^{M - 1}$.
\item [\rm{(3)}] When $a \neq 0, b = 0$, we have 
\begin{linenomath}\begin{align*}
& N(a, b, c; \beta) = N(a, b; \beta - c) \\
& = \left\{\begin{array}{ll}
    		q^{M - 1}\Big(1 + \epsilon \upsilon(\beta - c) q^{-\frac{r_{Q}}{2}} \Big), & 2|r_{Q}, \\
    	q^{M - 1}\Big(1 + \epsilon \eta\left(- a (\beta - c)\right) q^{\frac{1 - r_{Q}}{2}}\Big),  & 2\not|r_{Q}.
    	\end{array}
    	\right.
\end{align*}\end{linenomath}
\end{enumerate}
\end{corollary}

\section{Constructing codes with few weights from quadratic forms}

Let $\omega_0=0,\omega_1,\cdots,\omega_{q-1}$ be all the elements of $\F_q$ listed in some fixed order. For a vector $c = (c_0, c_1, \cdots, c_{n - 1}) \in \F_q^n$, the complete weight enumerator $\omega[c]$ of $c$ is defined as the formal expression 
\[\omega[c] = \omega_0^{k_0} \omega_1^{k_1} \cdots \omega_{q - 1}^{k_{q - 1}},\]
where $\sum\limits_{j = 0}^{q - 1} k_j = n, k_j$ is the number of components of $c$ that equal to $\omega_j$. The complete weight enumerator of a linear code $\mathcal{C}$ is defined by 
\[\omega[\mathcal{C}] = \sum\limits_{c \in \mathcal{C}} \omega[c].\]

In this section, we determine the complete weight enumerators and the weight hierarches of two classes
of linear codes.

\subsection{The code \texorpdfstring{$\mathcal{C}_{Q}$}.} 

 Recall that 
\begin{linenomath}\begin{equation}\label{eq:cq1} \mathcal{C}_{Q} = \{(a Q(x) + \Tr_{q_{2}/q}(b y))_{(x,y) \in \F^\star}: (a,b) \in \F_q \times \F_{q_{2}}\},\end{equation}\end{linenomath}
where $Q: \E \to \mathbb{F}_{q}$ is any arbitrary non-zero quadratic form.  

Using Lemma~\ref{lem:main2}, we obtain the following theorem. 

\begin{theorem}\label{thm:wd-eo}  Assume $Q(x)$ is a quadratic form on $\E/\mathbb{F}_{q}$ and $r_{Q}(\geq 1)$ is its rank. We have the following.

(1) If $r_{Q}$ is even,
then the code $\mathcal{C}_{Q}$ defined by \eqref{eq:cq1} is a $[q^M - 1, m_2 + 1]_{q}$ linear code
with the weight distribution in Table 1 and its complete weight enumerator is
\begin{linenomath}\begin{align*}	
CWE(\mathcal{C}_{Q}) &= \omega_0^{q^M - 1} + q (q^{m_2}-1) \omega_0^{q^{M - 1} - 1} \prod\limits_{\rho = 1}^{q - 1} \omega_\rho^{q^{M - 1}}\\		
& + (q - 1) \omega_0^{q^{M - 1}\left(1 + \epsilon q^{- \frac{r_{Q}}{2}}(q - 1)\right) - 1}\prod\limits_{\rho = 1}^{q - 1} \omega_\rho^{q^{M - 1}\left(1 - \epsilon q^{-\frac{r_{Q}}{2}}\right)}.
		\end{align*}\end{linenomath}
\begin{table}
\centering
\caption{Weight distribution of $\mathcal{C}_{Q}$ when $r_{Q}$ is even}
\begin{tabular*}{10.5cm}{@{\extracolsep{\fill}}ll}
\hline
\textrm{Weight} $\gamma$ \qquad& \textrm{frequency} $A_\gamma$   \\
\hline
0 & 1  \\
$q^{M - 1}(q - 1)$& $q (q^{m_2} - 1)$  \\
$q^{M - 1}(q - 1)\left(1 - \epsilon q^{- \frac{r_{Q}}{2}}\right)$ & $q - 1$  \\
\hline
\end{tabular*}
\end{table}

(2) If $r_{Q}$ is odd, then the code $\mathcal{C}_{Q}$ defined by \eqref{eq:cq1} is a $[q^M - 1, m_2 + 1]_{q}$ linear code with the weight distribution in Table 2 
and its complete weight enumerator is  
\begin{linenomath}\begin{align*}	
CWE(\mathcal{C}_{Q}) &= \omega_0^{q^M - 1} + q (q^{m_2} - 1) \omega_0^{q^{M - 1} - 1} \prod\limits_{\rho = 1}^{q - 1} \omega_\rho^{q^{M - 1}}\\		
& + \frac{q - 1}{2} \omega_0^{q^{M - 1} - 1} \prod\limits_{\rho = 1}^{q - 1} \omega_\rho^{q^{M - 1} \left(1 + \epsilon q^{\frac{1 - r_{Q}}{2}} \eta(- \omega_{\rho})\right)}\\
& + \frac{q - 1}{2} \omega_0^{q^{M - 1} - 1} \prod\limits_{\rho = 1}^{q - 1} \omega_\rho^{q^{M - 1} \left(1 - \epsilon q^{\frac{1 - r_{Q}}{2}} \eta(- \omega_{\rho})\right)}.
		\end{align*}\end{linenomath}
	Moreover, $\mathcal{C}_{Q}$ satisfies the Griesmer bound if and only if $m_1 = r_Q = 1$.
\begin{table}\label{2}
\centering
\caption{Weight distribution of $\mathcal{C}_{Q}$ when $r_{Q}$ is odd }
\begin{tabular*}{9cm}{@{\extracolsep{\fill}}ll}
\hline
\textrm{Weight} $\gamma$ & \textrm{frequency} $A_\gamma$   \\
\hline
$0$ & $1$  \\
$q^{M - 1}(q - 1)$ & $q^{m_2 + 1} - 1$\\
\hline
\end{tabular*}
\end{table}
\end{theorem}

\begin{proof} It is easy to check that the length of $\mathcal{C}_{Q}$ is $q^M - 1$ 
and the Hamming weight
of the codeword $c_{(a, b)} = (a Q(x) + \Tr_{q_{2}/q}(b y))_{(x,y) \in \F^\star}$ in $\mathcal{C}_{Q}$ is
\begin{equation}\label{eq:wt0}
     \mathrm{wt}(c_{(a, b)}) = q^{M} - N(a, b;0),
\end{equation} 
where $N(a, b;0)$ is defined in Lemma~\ref{lem:main2}. 

\rm{(1)} Assume $r_{Q}$ is even. From Lemma~\ref{lem:main2}, \eqref{eq:wt0} becomes 
\begin{linenomath}\begin{align*}
  \mathrm{wt}(c_{(a, b)}) = \left\{\begin{array}{ll}
    0,& a = b = 0,\\
    q^{M - 1}(q - 1),& b \neq 0,\\
    q^{M - 1}(q - 1)(1 - \epsilon q^{- \frac{r_{Q}}{2}}), & a \neq 0, b = 0. \\
    	\end{array}
    	\right.
\end{align*}\end{linenomath}
That is, $C_{Q}$ may have two nonzero weights $$\gamma_1 = q^{M - 1}(q - 1),\gamma_2 = q^{M - 1}(q - 1)(1 - \epsilon q^{- \frac{r_{Q}}{2}}),$$
and the multiplicities $A_{\gamma_i}$ of codewords with weight $\gamma_i$ in $C_Q$ is 
\begin{linenomath}\begin{align*}
  A_{\gamma_{1}} = \#\{(a, b) \in \mathbb{F}_{q} \times \mathbb{F}_{q_2} | b \neq 0 \} = q (q^{m_2} - 1),\\
  A_{\gamma_{2}} = \#\{(a, b) \in \mathbb{F}_{q} \times \mathbb{F}_{q_2} | a \neq 0, b = 0\} = q - 1.
\end{align*}\end{linenomath}
Noticing that for any nonzero $(a, b) \in \mathbb{F}_{q} \times \mathbb{F}_{q_2}$ the weight of the codeword $c_{(a, b)}$ is always nonzero, which means the dimension of $\mathcal{C}_{Q}$ is $ m_{2} + 1$.

It is easy to see that the complete weight enumerator $\omega[c_{(a, b)}]$ of the codeword $c_{(a, b)}$ is 
\[\omega_0^{N(a, b; 0)-1} \omega_1^{N(a, b; \omega_1)} \cdots \omega_{q - 1}^{N(a, b; \omega_{q - 1})}.\]
By Lemma~\ref{lem:main2} we have $\omega[c_{(a, b)}] = \omega_0^{q^M - 1}$ if $a = b = 0$ and $\omega_0^{q^{M - 1} - 1} \prod\limits_{\rho = 1}^{q - 1} \omega_\rho^{q^{M - 1}}$ if $b \neq 0$ and $\omega_0^{q^{M - 1}\left(1 + \epsilon q^{- \frac{r_Q}{2}}(q - 1)\right) - 1} \prod\limits_{\rho = 1}^{q - 1} \omega_\rho^{q^{M - 1} \left(1 - \epsilon q^{- \frac{r_{Q}}{2}}\right)}$ if $a \neq 0$. 
Then the desired complete weight enumerator of the code $C_Q$ 
can be obtained.

\rm{(2)} Assume $r_{Q}$ is odd. From Lemma~\ref{lem:main2}, \eqref{eq:wt0} becomes 
\begin{linenomath}\begin{align*}
  \mathrm{wt}(c_{(a, b)}) = \left\{\begin{array}{ll}
    0,&a = b = 0,\\
    q^{M - 1} (q - 1), &else, \\
    	\end{array}
    	\right.
\end{align*}\end{linenomath}
that is, $C_{Q}$ has only one nonzero weight $$\gamma = q^{M - 1}(q - 1),$$
and the multiplicities $A_{\gamma}$ of codewords with weight $\gamma$ in $C_Q$ can be easily obtained.
From the same reason, the dimension of $\mathcal{C}_{Q}$ is also $ m_{2} + 1$.

In this case, similarly,  by Lemma~\ref{lem:main2},  
\rm{(2a)} if $a = b = 0$, then $\omega[c_{(a, b)}] = \omega_0^{q^M - 1}$.\\ 
\rm{(2b)} If $b \neq 0$, then $\omega[c_{(a, b)}] = \omega_0^{q^{M - 1} - 1} \prod\limits_{\rho = 1}^{q - 1} \omega_\rho^{q^{M - 1}}$. \\
\rm{(2c)} If $a \neq 0, b = 0$ and $\eta(a) = 1$, then \[\omega[c_{(a, b)}] = \omega_0^{q^{M - 1} - 1} \prod\limits_{\rho = 1}^{q - 1} \omega_\rho^{q^{M - 1} \left(1 + \epsilon q^{\frac{1 - r_Q}{2}} \eta(-\omega_{\rho})\right)}\]
\rm{(2d)} If $a \neq 0, b = 0$ and $\eta(a) = - 1$, then  \[\omega[c_{(a, b)}] = \omega_0^{q^{M - 1} - 1} \prod\limits_{\rho = 1}^{q-1} \omega_\rho^{q^{M - 1}\left(1 - \epsilon q^{\frac{1 - r_Q}{2}}\eta(-\omega_{\rho})\right)}.\]
Thus the complete weight enumerator of the code $C_Q$ can be desired.

Now we check the optimality. In this case, the minimum distance of $\mathcal{C}_{Q}$ is $d=q^{M-1} (q - 1)$, which means the code is a $[q^M-1,m_2 + 1,q^{M - 1}(q - 1)]_{q}$ linear code. Thus 
\begin{linenomath}\begin{align*}
   \sum\limits_{i = 0}^{m_2} \lceil \frac{d}{q^i} \rceil 
  & = (q-1)(q^{M-1} + q^{M-2} + \cdots + q^{m_1-1})   \\
  &= q^{M} - q^{m_1-1},
\end{align*}\end{linenomath}
which means that the code
 $\mathcal{C}_{Q}$ is optimal with respect to the Griesmer bound if and only if $m_1 = 1$.

\end{proof}

\begin{lemma}[Ashikhmin–Barg lemma {\cite{AABA1998, YD2006}}]\label{lem:minimal} 
A linear code $\mathcal{C}$ over $\F_q$ is minimal if 
\[
\frac{w_{min}}{w_{max}} > \frac{q-1}{q}, 
\]
where $w_{min}$ and $w_{max}$ denote the minimum and maximum nonzero Hamming weights in the code $\mathcal{C}$.
\end{lemma}

From Lemma~\ref{lem:minimal} and Theorem~\ref{thm:wd-eo} we can obtain the follow corollary.

\begin{corollary}
Assume $Q(x)$ is a quadratic form on $\E/\mathbb{F}_{q}$ with the rank $r_{Q} \geq 1$ and $\mathcal{C}_Q$ is the code defined in \eqref{eq:cq1}.
\begin{enumerate}
    \item [\rm{(1)}]When $r_{Q}$ is even, the code $\mathcal{C}_Q$ is minimal if $\epsilon = 1$ and $r_Q > 2$ or $\epsilon = - 1$.
    \item [\rm{(2)}]When $r_{Q}$ is odd, the code $\mathcal{C}_Q$ is minimal.
    
\end{enumerate}
\end{corollary}

In the following sequel we shall determine the weight hierarchy of $\mathcal{C}_{Q}$.
By Theorem~\ref{thm:wd-eo}, we know that the dimension of the code $\mathcal{C}_{Q}$ defined by \eqref{eq:cq1} is $m_2 + 1$. So, the map $\phi: \mathbb{F}_{q} \times \mathbb{F}_{q_{2}} \rightarrow \mathcal{C}_{Q}$, 
defined by \[\phi(a,b) = c_{(a,b)} = (a Q(x) + \Tr_{q_2 / q}(b y))_{(x,y)\in \F^\star},\]
is an $\mathbb{F}_{q}$-linear isomorphism. 
Under this isomorphism, the set of $r$-dimensional subspaces of $\mathcal{C}_{Q}$ bijectively corresponds to those of $\mathbb{F}_{q} \times \mathbb{F}_{q_{2}}$:
\[[\mathcal{C}_{Q}, r]_{q} = \{\phi(H_r): H_r \in [\mathbb{F}_{q} \times \mathbb{F}_{q_{2}}, r]_{q}\}.\]
Consequently, for any $ K \in [\mathcal{C}_{Q}, r]_{q}$, there exists a unique $H_{r} \in [\mathbb{F}_{q} \times \mathbb{F}_{q_{2}}, r]_{q}$ such that $K = \phi(H_{r})$. 
Thus $\#\textrm{Supp}(K)$ is computed as follows: 
\begin{linenomath}\begin{align*}
    \#\textrm{Supp}(K) & = \#\textrm{Supp}(\phi(H_r))\\
    & = \# \Big\{(x, y) \in \F^\star: a Q(x) + \Tr_{q_2/q}(b y) \ne 0 \ \ \textrm{for some}\ \  (a, b) \in H_r  \Big\}\\
    & = q^M - 1 - \#\Big\{(x, y) \in \F^\star: a Q(x) + \Tr_{q_2/q}(b y) = 0, \forall (a, b) \in H_r \Big\}.
\end{align*}\end{linenomath}
For any $H_{r} \in [\mathbb{F}_{q} \times \mathbb{F}_{q_{2}}, r]_{q}$, define \begin{equation}\label{eq:cqnr}N(H_r) = \#\{(x,y) \in \mathbb{F}^\star:aQ(x) + \Tr_{q_2/q}(b y) = 0,  \forall (a,b) \in H_r \}.\end{equation}
This directly yields the 
$r$-th generalized Hamming weight
\begin{linenomath}\begin{align}\label{eq:gwt1}
 d_{r}(\mathcal{C}_{Q}) & = q^M - 1 - \max\Big\{N(H_r): H_r \in [\mathbb{F}_{q} \times \mathbb{F}_{q_{2}},r]_{q}\Big\}. 
       \end{align}\end{linenomath}
Therefore, the determination of $d_{r}(\mathcal{C}_{Q})$ reduces to calculating $N(H_r)$.

\begin{lemma} \label{lem:gwt2} Assume $Q(x)$ is a quadratic form on $\E/\mathbb{F}_{q}$ with rank $r_{Q}$. Let $H_{r}$ be an $r$-dimensional subspace of $\mathbb{F}_{q} \times \mathbb{F}_{q_{2}}$ and $N(H_r)$ be defined by \eqref{eq:cqnr}. Then
$$
N(H_r) = \begin{cases}q^{M - r}\left(\epsilon t q^{- \frac{r_{Q}}{2}}  + 1\right) - 1, & 2\mid r_{Q},\\
q^{M - r} - 1, & 2\nmid r_{Q}\end{cases}
$$
where $t = \#\{(a, 0) \in H_r: a \ne 0\}$.
\end{lemma}
\begin{proof} By the orthogonal property of additive characters, we have
\begin{linenomath}\begin{align*}
q^{r} (N(H_r) + 1) &= \sum\limits_{(x, y) \in \mathbb{F}} \sum\limits_{(a, b) \in H_r}\zeta_{p}^{\Tr_{q/p}(a Q(x) + \Tr_{q_2/q}(b y))}  \\ \nonumber
& = \sum\limits_{(a, b) \in H_r} \sum\limits_{x \in \mathbb{F}_{q_1}}\zeta_{p}^{\Tr_{q/p}(a Q(x))} \sum\limits_{y \in \mathbb{F}_{q_2}} \zeta_{p}^{\Tr_{q_2/p}(b y)}   \\ \nonumber
& =  q^{m_2}\sum\limits_{(a, 0) \in H_r} \sum\limits_{x \in \mathbb{F}_{q_1}}\zeta_{p}^{\Tr_{q/p}(a Q(x))}    \\ \nonumber
&= q^{m_2}\sum\limits_{\substack{(a, 0) \in H_r\\
a \ne 0}}\sum\limits_{x \in \mathbb{F}_{q_1}}\zeta_{p}^{\Tr_{q/p}(a Q(x))} + q^{m_1} q^{m_2},
\end{align*}\end{linenomath}
where the third equality follows from the orthogonality of when $b \neq 0$.
Applying Lemma~\ref{lem:main1}, the desired conclusion is obtained.
\end{proof}

Obviously, for $H_{r}$ and $t$ defined in Lemma~\ref{lem:gwt2}, if $ (1, 0) \in H_{r}$, then $t = q - 1$; if $ (1, 0) \notin H_{r}$, then $t = 0$. Thus, by formula~\eqref{eq:gwt1} and Lemma~\ref{lem:gwt2}, we have the following theorem.

\begin{theorem}\label{thm:gwt1} Assume $Q(x)$ is a quadratic form on $\E/\mathbb{F}_{q}$ and $\mathcal{C}_{Q}$ defined by \eqref{eq:cq1}. 
For each $ r $ and $ 1\leq r \leq m_2 + 1$, we have the following.\\
   (1) When $r_{Q}$ is even and $\epsilon = 1$,
    \begin{linenomath}\begin{align*}
    d_r(\mathcal{C}_{Q}) = q^{M - r}\left(q^{r} - 1 - (q - 1) q^{-\frac{r_{Q}}{2}}\right).
\end{align*}\end{linenomath}
(2) When $r_{Q}$ is even and $\epsilon = - 1$,
\begin{linenomath}\begin{align*}
    &d_r(\mathcal{C}_{Q}) = \begin{cases}
   q^{M - r}\left(q^{r} - 1\right), & 1 \leq r < m_{2} + 1,\\
    q^{M - r}\left(q^{r} - 1 + (q - 1) q^{-\frac{r_{Q}}{2}}\right), & r = m_{2} + 1.\\
     \end{cases}
\end{align*}\end{linenomath}
(3) When $r_{Q}$ is odd,
\begin{linenomath}\begin{align*}
    d_r(\mathcal{C}_{Q}) = q^{M - r}\left(q^{r} - 1\right).
\end{align*}\end{linenomath}
\end{theorem}

As special cases of Theorem~\ref{thm:wd-eo} and Theorem~\ref{thm:gwt1}, we give the following four examples, which are verified by Magma programs.

\begin{example}
 Let $(q, m_{1}, m_{2}) = (3, 4, 3)$, $Q(x)=\Tr_{3^4/3}(x^{2})$,
 we have $r_{Q} = 4$ and $\varepsilon_Q = -1$. Then, the corresponding code $\mathcal{C}_{Q}$ has parameters $[2186, 4, 1458]_{3}$ 
 and the complete weight enumerator is 
\begin{linenomath}\begin{align*}
  \omega_{0}^{2186} + 78 \omega_{0}^{728} \omega_{1}^{729} \omega_{2}^{729} + 2 \omega_{0}^{566} \omega_{1}^{810}\omega_{2}^{810}.\end{align*}\end{linenomath}
And the weight hierarchy is $$ d_1(\mathcal{C}_{Q})= 1458, d_2(\mathcal{C}_{Q})= 1944, d_3(\mathcal{C}_{Q})=2106 , d_4(\mathcal{C}_{Q})= 2166.$$
\end{example}

\begin{example}
 Let $(q, m_{1}, m_{2}) = (5, 3, 2)$, $Q(x) = \Tr_{5^3/5}(x^{2}) - \frac{1}{3}\Tr_{5^3/5}(x)^2$,
 we have $r_{Q} = 2$ and $\varepsilon_Q = - 1$. Then, the corresponding code $\mathcal{C}_{Q}$ has parameters $[3124, 3, 2500]_{5}$ 
 and the complete weight enumerator is 
\begin{linenomath}\begin{align*}
  \omega_{0}^{3124} + 120 \omega_{0}^{624} \omega_{1}^{625} \omega_{2}^{625} \omega_{3}^{625} \omega_{4}^{625} + 4 \omega_{0}^{124} \omega_{1}^{750} \omega_{2}^{750} \omega_{3}^{750} \omega_{4}^{750}.\end{align*}\end{linenomath}
And the weight hierarchy is $$ d_1(\mathcal{C}_{Q})= 2500, d_2(\mathcal{C}_{Q})= 3000, d_3(\mathcal{C}_{Q})= 3120.$$
\end{example}

\begin{example}
 Let $(q, m_{1}, m_{2}) = (9, 3, 2), Q(x) = \Tr_{9^3/9}(g x^{2})$
 , where $g$ is a primitive element of $\mathbb{F}_{9^3}$, we have $r_{Q} = 3$ and $\varepsilon_Q = - 1$. Also, we assume $\eta(- \omega_i) = 1, i = 1, \ldots, 4$ and $\eta(- \omega_j) = - 1, i = 5, \ldots, 8$   Then, the corresponding code $\mathcal{C}_{Q}$ has parameters $[59048, 3, 52488]_{9}$ 
 and the complete weight enumerator is 
\begin{linenomath}\begin{align*}
  &\omega_{0}^{59048} + 720 \omega_{0}^{6560} \omega_{1}^{6561} \omega_{2}^{6561} \omega_{3}^{6561} \omega_{4}^{6561} \omega_{5}^{6561} \omega_{6}^{6561} \omega_{7}^{6561} \omega_{8}^{6561}\\
  & + 4 \omega_{0}^{6560} \omega_{1}^{5832} \omega_{2}^{5832} \omega_{3}^{5832} \omega_{4}^{5832} \omega_{5}^{7290} \omega_{6}^{7290} \omega_{7}^{7290} \omega_{8}^{7290}\\
  & + 4 \omega_{0}^{6560} \omega_{1}^{7290} \omega_{2}^{7290} \omega_{3}^{7290} \omega_{4}^{7290} \omega_{5}^{5832} \omega_{6}^{5832} \omega_{7}^{5832} \omega_{8}^{5832}.\end{align*}\end{linenomath}
And the weight hierarchy is $$ d_1(\mathcal{C}_{Q})=52488, d_2(\mathcal{C}_{Q})=52830 , d_3(\mathcal{C}_{Q})=58968 .$$
\end{example}

\begin{example}
 Let $(q, m_{1}, m_{2}) = (5, 2, 3), Q(x) = \Tr_{5^2/5}(x^{2}) - \frac{1}{2} \Tr_{5^2/5}(x)^2$,
 we have $r_{Q} = 1$ and $\varepsilon_Q = 1$. And assume $\eta(- \omega_1) = 1, \eta(- \omega_2) = 1, \eta(- \omega_3) = -1, \eta(- \omega_4) = - 1$.Then, the corresponding code $\mathcal{C}_{Q}$ has parameters $[3124, 4, 2500]_{5}$ 
 and the complete weight enumerator is 
\begin{linenomath}\begin{align*}
  \omega_{0}^{3124} + 620 \omega_{0}^{624} \omega_{1}^{625} \omega_{2}^{625} \omega_{3}^{625} \omega_{4}^{625} + 2 \omega_{0}^{624} \omega_{1}^{1250} \omega_{2}^{1250} + 2 \omega_{0}^{624} \omega_{3}^{1250} \omega_{4}^{1250}.\end{align*}\end{linenomath}
And the weight hierarchy is $$ d_1(\mathcal{C}_{Q})=2500, d_2(\mathcal{C}_{Q})=3000, d_3(\mathcal{C}_{Q})=3100 , d_4(\mathcal{C}_{Q})=3120 .$$
\end{example}

\subsection{The code \texorpdfstring{$\mathcal{C}_{Q}'$}.} 

 Recall that 
 \begin{linenomath}\begin{equation}\label{eq:cq2}\mathcal{C}_{Q}' = \{(a Q(x) + \Tr_{q_{2}/q}(b y) + c)_{(x,y) \in \F}: (a,b,c) \in \F_q \times \F_{q_{2}} \times \F_q\},\end{equation}\end{linenomath}
where $Q: \E \to \mathbb{F}_{q}$ is any arbitrary non-zero quadratic form.

From Corollary~\ref{pro:main3}, we can directly obtain the weights and frequencies of $\mathcal{C}_{Q}'$ in the following theorem. We omit its proof.

\begin{theorem}\label{thm:eo1} Assume $Q(x)$ is a quadratic form on $\E/\mathbb{F}_{q}$ and $r_{Q}$ is its rank. We have the following.

(1) If $r_{Q}$ is even,
then the code $\mathcal{C}_{Q}'$ defined by \eqref{eq:cq2} is a $[q^M , m_2 + 2]_{q}$ linear code
with the weight distribution in Table 3 and its complete weight enumerator is
\begin{linenomath}\begin{align*}	
CWE(\mathcal{C}_{Q}') = \sum\limits_{i = 0}^{q - 1} \omega_i^{q^M} + q^2(q^{m_2} - 1) \prod\limits_{\rho = 0}^{q - 1} \omega_\rho^{q^{M - 1}} + \sum\limits_{i = 0}^{q - 1}(q - 1)\omega_i^{t_1}\prod\limits_{\rho = 0, \rho \ne i}^{q - 1} \omega_\rho^{t_2},
\end{align*}\end{linenomath}
where 
\begin{linenomath}\begin{align*}
    & t_1 = q^{M - 1} \left(1 + \epsilon (q - 1) q^{- \frac{r_{Q}}{2}}\right), t_2 = q^{M - 1}\left(1 - \epsilon q^{- \frac{r_{Q}}{2}}\right).
    \end{align*}\end{linenomath}
\begin{table}
\centering
\caption{Weight distribution of $\mathcal{C}_{Q}'$ when $r_{Q}$ is even}
\begin{tabular*}{10.5cm}{@{\extracolsep{\fill}}ll}
\hline
\textrm{Weight} $\gamma$ \qquad& \textrm{frequency} $A_\gamma$   \\
\hline
0 & 1  \\
$q^M$ & $q - 1$  \\
$q^{M - 1} (q - 1)$ & $q^2 (q^{m_2} - 1)$  \\
$q^{M - 1} (q - 1)(1 - \epsilon q^{-\frac{r_{Q}}{2}})$ & $q - 1$  \\
$q^{M - 1} (q - 1 + \epsilon q^{- \frac{r_{Q}}{2}})$ & $(q - 1)^2$  \\
\hline
\end{tabular*}
\end{table}
(2) If $r_{Q}$ is odd, then the code $\mathcal{C}_{Q}'$ defined by \eqref{eq:cq2} is a $[q^M , m_2 + 2]_{q}$ linear code
with the weight distribution in Table 4 and its complete weight enumerator is
\begin{linenomath}\begin{align*}	
CWE(\mathcal{C}_{Q}') &= \sum\limits_{i = 0}^{q - 1} \omega_i^{q^M} + q^2 (q^{m_2} - 1) \prod\limits_{\rho = 0}^{q - 1} \omega_\rho^{q^{M - 1}}\\
    & + \sum\limits_{i = 0}^{q - 1} \frac{1}{2}(q-1) \omega_i^{q^{M - 1}} \prod\limits_{\substack{\rho = 0,\\ \rho \ne i}}^{q - 1} \omega_\rho^{k_1} + \sum\limits_{i = 0}^{q - 1} \frac{1}{2} (q - 1)\omega_i^{q^{M - 1}}\prod\limits_{\substack{\rho = 0,\\ \rho \ne i}}^{q - 1} \omega_\rho^{k_2},
\end{align*}\end{linenomath}
where 
\begin{linenomath}\begin{align*}
    & k_1 = q^{M - 1} \left(1 + \epsilon q^{\frac{1 - r_{Q}}{2}} \eta(\omega_i - \omega_\rho)\right), k_2 = q^{M - 1} \left(1 - \epsilon q^{\frac{1 - r_{Q}}{2}} \eta(\omega_i - \omega_\rho)\right).
    \end{align*}\end{linenomath}
\begin{table}
\centering
\caption{Weight distribution of $\mathcal{C}_{Q}'$ when $r_{Q}$ is odd}
\begin{tabular*}{10.5cm}{@{\extracolsep{\fill}}ll}
\hline
\textrm{Weight} $\gamma$ \qquad& \textrm{frequency} $A_\gamma$   \\
\hline
$0$ & $1$  \\
$q^M$ & $q - 1$  \\
$q^{M - 1} (q - 1)$ & $q^2 (q^{m_2} - 1) + q - 1$  \\
$ q^{M - 1} \left(q - 1 - \epsilon q^{\frac{1 - r_{Q}}{2}}\right)$ & $\frac{1}{2} (q - 1)^2$  \\
$q^{M - 1} \left(q - 1 + \epsilon q^{\frac{1 - r_{Q}}{2}}\right)$ & $\frac{1}{2} (q - 1)^2$  \\
\hline
\end{tabular*}
\end{table}
\end{theorem}

In the following sequel we shall determine the weight hierarchy of $\mathcal{C}_{Q}'$.

\begin{theorem}\label{thm:code2:whh} Assume $Q(x)$ is a quadratic form on $\E/\mathbb{F}_{q}$ and $\mathcal{C}_{Q}'$ defined by \eqref{eq:cq2}. For each $ r $ and $ 1\leq r \leq m_2 + 2$, we have the following.\\
   (1) When $r_{Q}$ is even and $\epsilon = 1$,
    \begin{linenomath}\begin{align*}
    d_r(\mathcal{C}_{Q}') = \begin{cases}
    q^{M - r} \left(q^{r} - 1 - (q - 1) q^{-\frac{r_{Q}}{2}}\right), & 1 \leq r < m_{2} + 2,\\
    q^M, & r = m_{2} + 2.\\
     \end{cases} \end{align*}\end{linenomath}
(2) When $r_{Q}$ is even and $\epsilon = - 1$,
\begin{linenomath}\begin{align*}
    &d_r(\mathcal{C}_{Q}') = \begin{cases}
    q^{M - r} \left(q^{r} - 1 - q^{- \frac{r_{Q}}{2}}\right), & 1 \leq r < m_{2} + 2,\\
    q^M, & r = m_{2} + 2.\\
     \end{cases}
\end{align*}\end{linenomath}
(3) When $r_{Q}$ is odd,
\begin{linenomath}\begin{align*}
    d_r(\mathcal{C}_{Q}') = \begin{cases}
   q^{M - r}\left(q^{r} - 1 - q^{\frac{1 - r_{Q}}{2}}\right), & 1 \leq r < m_{2} + 2,\\
    q^M, & r = m_{2} + 2.\\
     \end{cases}
\end{align*}\end{linenomath}
\end{theorem}
\begin{proof} 
Similar to \eqref{eq:gwt1}, we proved that the 
$r$-th generalized Hamming weight of $\mathcal{C}_{Q}'$ is \begin{linenomath}\begin{align}\label{eq:gwt2}
 d_{r}(\mathcal{C}_{Q}') & = q^M - \max\Big\{N(H_r'): H_r'\in [\mathbb{F}_{q} \times \mathbb{F}_{q_{2}} \times \mathbb{F}_{q},r]_{q}\Big\}, 
       \end{align}\end{linenomath}
where $N(H_r')$ is defined by \begin{equation*}N(H_r') = \#\{(x,y) \in \mathbb{F}:aQ(x) + \Tr_{q_2/q}(b y) + c = 0,  \forall (a,b,c) \in H_r' \}.\end{equation*}
Thus it suffices to compute  $N(H_r')$. 

By the orthogonal property of additive characters, we have
\begin{linenomath}\begin{align}\label{eq:gw3}
N(H_r') & = q^{- r} \sum\limits_{(x, y) \in \mathbb{F}} \sum\limits_{(a, b, c) \in H_{r}^{\prime}} \zeta_p^{\Tr_{q/p}\left(a Q(x) + \Tr_{q_{2}/q}(b y) + c\right)}\\ \nonumber
& = q^{- r} \sum\limits_{(a, b, c) \in H_{r}'} \sum\limits_{(x, y) \in \mathbb{F}}\zeta_p^{\Tr_{q/p}\left(a Q(x) + \Tr_{q_{2}/q}(b y) + c\right)}\\ \nonumber
& = q^{- r} \sum\limits_{(a, b, c) \in H_r'}\zeta_p^{\mathrm{Tr}_{q/p}(c)} \sum\limits_{x \in \mathbb{F}_{q_1}} \zeta_p^{\Tr_{q/p}(a Q(x))}\sum\limits_{y \in\mathbb{F}_{q_2}}\zeta_p^{\Tr_{q_{2}/p}(b y)}\\ \nonumber
& = q^{m_{2} - r} \sum\limits_{(a, 0, c) \in H_r'} \zeta_p^{\Tr_{q/p}(c)} \sum\limits_{x \in\mathbb{F}_{q_1}}\zeta_p^{\Tr_{q/p}(a Q(x))}\\ \nonumber
& = q^{m_{2} - r}  \sum\limits_{\substack{(a, 0, c) \in H_{r}'\\ a \ne 0}} \zeta_{p}^{\Tr_{q/p}(c)} \sum\limits_{x \in \mathbb{F}_{q_1}} \zeta_p^{\Tr_{q/p}(a Q(x))} + q^{M - r} \sum\limits_{(0, 0, c) \in H_r'}\zeta_p^{\Tr_{q/p}(c)}.\nonumber
\end{align}\end{linenomath}

(1) When $r_{Q}$ is even and $\epsilon = 1$, from Lemma~\ref{lem:main1}, the equation~\eqref{eq:gw3} becomes
\begin{linenomath}\begin{align*}
N(H_r') &  = q^{M - r - \frac{r_{Q}}{2}} \sum\limits_{\substack{(a, 0, c) \in H_{r}'\\ a \ne 0}} \zeta_{p}^{\Tr_{q/p}(c)} + q^{M - r} \sum\limits_{(0, 0, c) \in H_{r}^{'}} \zeta_{p}^{\Tr_{q/p}(c)}\\
&= q^{M - r - \frac{r_{Q}}{2}}  (\sum\limits_{\substack{(a, 0, 0) \in H_{r}'\\ a \ne 0}} 1 + \sum\limits_{\substack{(a, 0, c) \in H_{r}'\\ ac \ne 0}} \zeta_{p}^{\Tr_{q/p}(c)}) + q^{M - r}(1 + \sum\limits_{\substack{(0, 0, c) \in H_{r}^{'}\\ c\ne 0}} \zeta_{p}^{\Tr_{q/p}(c)})\\
&=q^{M - r - \frac{r_{Q}}{2}}  (\sum\limits_{\substack{(a, 0, 0) \in H_{r}'\\ a \ne 0}} 1 + \frac{1}{q-1}\sum\limits_{\substack{(a, 0, c) \in H_{r}'\\ ac \ne 0}} \sum_{\lambda\in\F_q^*}\zeta_{p}^{\Tr_{q/p}(\lambda c)}) \\
&+ q^{M - r}(1 + \frac{1}{q-1}\sum\limits_{\substack{(0, 0, c) \in H_{r}^{'}\\ c\ne 0}}\sum_{\lambda\in\F_q^*} \zeta_{p}^{\Tr_{q/p}(\lambda c)})\\
&=q^{M - r - \frac{r_{Q}}{2}}  (\sum\limits_{\substack{(a, 0, 0) \in H_{r}'\\ a \ne 0}} 1 - \frac{1}{q-1}\sum\limits_{\substack{(a, 0, c) \in H_{r}'\\ ac \ne 0}} 1) + q^{M - r}(1 - \frac{1}{q-1}\sum\limits_{\substack{(0, 0, c) \in H_{r}^{'}\\ c\ne 0}} 1).
\end{align*}\end{linenomath}
It is easy to see that when $1 \leq r < m_{2} + 2$, $N(H_r')$ could reach its maximum $q^{M - r} \left((q - 1) q^{- \frac{ r_{Q}}{2}} + 1\right)$ and 
$$d_r(\mathcal{C}_{Q}') = q^{M - r} \left(q^r - 1 - (q - 1) q^{-\frac{r_{Q}}{2}}\right).$$
When $r = m_{2} + 2$, $N(H_r')$ only has one value $0$ and 
$$d_r(\mathcal{C}_{Q}') = q^{M}.$$

(2) When $r_{Q}$ is even and $\epsilon = - 1$, from Lemma~\ref{lem:main1} and a similar discussion, the equation~\eqref{eq:gw3} becomes
\begin{linenomath}\begin{align*}
N(H_r') = q^{M - r - \frac{r_{Q}}{2}}  (\frac{1}{q-1}\sum\limits_{\substack{(a, 0, c) \in H_{r}'\\ ac \ne 0}} 1 - \sum\limits_{\substack{(a, 0, 0) \in H_{r}'\\ a \ne 0}} 1) + q^{M - r}(1 - \frac{1}{q-1}\sum\limits_{\substack{(0, 0, c) \in H_{r}^{'}\\ c\ne 0}} 1).
\end{align*}\end{linenomath}
It is easy to see that when $1 \leq r < m_{2} + 2$, $N(H_r')$ could reach its maximum $q^{M - r} \left(1 + q^{- \frac{ r_{Q}}{2}}\right)$ and 
$$d_r(\mathcal{C}_{Q}') = q^{M - r} \left(q^r - 1 -  q^{-\frac{r_{Q}}{2}}\right).$$
When $r = m_{2} + 2$, $N(H_r')$ only has one value $0$ and 
$$d_r(\mathcal{C}_{Q}') = q^{M}.$$

(3) When $r_{Q}$ is odd, from Lemma~\ref{lem:main1} and Lemma~\ref{lem:basic}, the equation~\eqref{eq:gw3} becomes 
\begin{linenomath}\begin{align*}
N(H_r') 
&= \frac{1}{q - 1}q^{M - r} \left(- 1\right)^{m - 1} \varepsilon_{Q} (p^*)^{- \frac{m r_{Q}}{2}} \sum\limits_{\substack{(a, 0, c) \in H_{r}'\\ ac \neq 0}}\eta (- a/c)\sum_{\lambda\in\F_q^*} \zeta_{p}^{\Tr_{q/p} (\lambda c)} \eta (\lambda c)\\
& + q^{M - r} \sum\limits_{(0, 0, c) \in H_{r}'} \zeta_{p}^{\Tr_{q/p}(c)}\\
&=\frac{1}{q - 1}q^{M - r + \frac{1 - r_{Q}}{2}}\epsilon\sum\limits_{\substack{(a, 0, c) \in H_{r}'\\ ac \neq 0}}\eta (ac)+ q^{M - r} \sum\limits_{(0, 0, c) \in H_{r}'} \zeta_{p}^{\Tr_{q/p}(c)}.
\end{align*}\end{linenomath}
When $1 \leq r < m_{2} + 2$, $N(H_r')$ could reach its maximum $q^{M - r + \frac{1 - r_{Q}}{2}} + q^{M - r}$ and 
$$d_r(\mathcal{C}_{Q}') = q^{M - r} \left(q^r - 1 -  q^{\frac{1 - r_{Q}}{2}}\right).$$
When $r = m_{2} + 2$, $N(H_r')$ only has one value $0$ and 
$$d_r(\mathcal{C}_{Q}') = q^{M}.$$

\end{proof}

As special cases of Theorem~\ref{thm:eo1} and Theorem~\ref{thm:code2:whh}, we give the following two examples, which are verified by Magma programs.

\begin{example}\label{ex:353}
 Let $(q, m_{1}, m_{2}) = (3, 5, 3), Q(x) = \Tr_{3^5/3}(2 x^{10} + x^{2})$, we have $r_{Q} = 4$ and $\varepsilon_Q = - 1$. Then, the corresponding code $\mathcal{C}_{Q}'$ is a linear code with parameters $[6561, 5, 4131]_{3}$, has the complete weight enumerator 
              \begin{linenomath}\begin{align*}
&CWE(\mathcal{C}_{Q}') = \omega_0^{6561} + \omega_1^{6561} + \omega_2^{6561} + 234 \omega_0^{2187} \omega_1^{2187} \omega_2^{2187}\\
& + 2 \omega_0^{1701} \omega_1^{2430} \omega_2^{2430} + 2 \omega_0^{2430} \omega_1^{1701} \omega_2^{2430} + 2 \omega_0^{2430} \omega_1^{2430} \omega_2^{1701},
    \end{align*}\end{linenomath}
    and has the weight hierarchy $$d_1(\mathcal{C}_{Q}') = 4131, d_2(\mathcal{C}_{Q}') = 5741, d_3(\mathcal{C}_{Q}') = 6291, d_4(\mathcal{C}_{Q}') = 6471, d_5(\mathcal{C}_{Q}') = 6561.$$
\end{example}


\begin{example} Let $(q, m_{1}, m_{2}) = (3, 3, 4), Q(x) = \Tr_{3^3/3}(\theta x^{2})$, where $\theta$ is a primitive element of $\mathbb{F}_{3^3}$, by Corollary 1 in \cite{TXF2017}, we have $r_{Q} = 3$ and $\varepsilon_Q = - 1$. And assume $\eta(\omega_0) = 0, \eta(\omega_1) = 1, \eta(\omega_2) = - 1$. Then, the corresponding code $\mathcal{C}_{Q}'$ is a linear code with parameters $[2187, 6, 1215]_{3}$, has the complete weight enumerator 
 \begin{linenomath}\begin{align*}
     &CWE(\mathcal{C}_{Q}') = \omega_0^{2187} + \omega_1^{2187} + \omega_2^{2187} + 720 \omega_0^{729} \omega_1^{729} \omega_2^{729}\\
& + \omega_0^{729} \omega_1^{972} \omega_2^{486} + \omega_0^{486} \omega_1^{729} \omega_2^{972} + \omega_0^{972} \omega_1^{486} \omega_2^{729}\\
& + \omega_0^{729} \omega_1^{486} \omega_2^{972} + \omega_0^{972} \omega_1^{729} \omega_2^{486} + \omega_0^{486} \omega_1^{972} \omega_2^{729}, 
 \end{align*}\end{linenomath} 
 and has the weight hierarchy \begin{align*} & d_1(\mathcal{C}_{Q}') = 1215, d_2(\mathcal{C}_{Q}') = 1863, d_3(\mathcal{C}_{Q}') = 2079, d_4(\mathcal{C}_{Q}') = 2151, d_5(\mathcal{C}_{Q}') = 2175,\\ &d_6(\mathcal{C}_{Q}') = 2187.\end{align*}
\end{example}

\section{Descending \texorpdfstring{$\F_q$}.-codes to \texorpdfstring{$\F_p$}.-codes}

Xie et al \cite{XOYM2023} presented a trick to descend an $\mathbb{F}_{q}$-linear code to an $\mathbb{F}_{p}$-linear code, whose weight enumerator is uniquely determined by the old code, and consequently the optimality would be inherited. In this subsection we should consider the descended codes of $\mathcal{C}_{Q}$ and $\mathcal{C}_{Q}'$, and determine their weight hierarchies.

\begin{definition}[{\cite[Definition 5]{XOYM2023}}]
 Assume $q = p^m$ and $N$ is a factor of $p - 1$  prime to $\frac{q - 1}{p - 1}$. Let $\theta$ be a primitive $\frac{q - 1}{N}$-th root of unity in $\mathbb{F}_q$. Define the $\F_p$-linear map
$$
\Psi_N: \mathbb{F}_q \to \mathbb{F}_p^{\frac{q - 1}{N}} = \mathbb{F}_p^{\frac{q - 1}{N} \times 1}, \gamma \mapsto \psi_\gamma = \left(\mathrm{Tr}_{q/p}(\gamma \theta^i)\right)_{0 \leq i < \frac{q - 1}{N}}^{\mathrm{T}}.
$$
The code $\mathcal{C}_N := \mathrm{Im}\Psi_N$.
\end{definition}

By Proposition~3 \cite{XOYM2023}, $\Psi_N$ is injective and $\mathcal{C}_N$ is an $[\frac{q - 1}{N}, m, \frac{(p - 1)p^{m-1}}{N}]$ constant-weight code over $\F_p$.

\begin{lemma} 
Let $H = \Big\langle \theta \Big\rangle$ be the multiplicative subgroup of $\mathbb{F}_{q}^*$ generated by $\theta$ and 
consider the left multiplication action of $\F_p^*$ on the quotient group $\F_q^*/H$, that is,
$$
\lambda \cdot (cH) = (\lambda c)H, \quad \lambda \in \F_p^*.
$$ Then the action is transitive.
\end{lemma}

\begin{proof}
 It's easy to check that the stabilizer of a coset $cH$ is
$$
\mathrm{Stab}(cH) = \{ \lambda \in \F_p^* : \lambda c H = cH \} = \{ \lambda \in \F_p^* : \lambda \in H \} = \F_p^* \cap H.
$$
Define $d = |\F_p^* \cap H|$, by Lagrange theorem, we know that $d \mid \gcd(|\F_p^*|, |H|)$. By the condition $\gcd(N, \frac{q - 1}{p - 1}) = 1$ we have
\begin{align*}\gcd(|\F_p^*|, |H|) &= \gcd(p - 1, L) = \gcd(N \cdot \frac{p - 1}{N}, \frac{p^m - 1}{p-1} \cdot \frac{p - 1}{N}) \\
&= \frac{p -1}{N} \cdot \gcd(N, M) = \frac{p - 1}{N},\end{align*}
which concludes that $d\mid \frac{p -1}{N}$. Since $\F_q^*$ is a cyclic group of order $q - 1 = \frac{p - 1}{N} \cdot \frac{N(q - 1)}{p - 1}$, it has a unique subgroup $K$ of order $\frac{p - 1}{N}$. For the same reason, we have  $K \subseteq \F_p^* \cap H$
and thus $\frac{p - 1}{N} \mid d$. From the above discussion, we obtain that $d = \frac{p - 1}{N}$. By the orbit formula, we know that the orbit size is 
$$
|\mathrm{Orb}(cH)| = \frac{p - 1}{d} = N.
$$
That means there is only one orbit, and the action is transitive.
\end{proof}

\begin{corollary}\label{cor:le:s1}
For any $c \in \F_q^*$, we have
\begin{align*}
\sum_{\lambda \in \F_p^*} \sum_{i=0}^{\frac{q - 1}{N} - 1} \zeta_p^{\Tr_{q/p}(\lambda c \theta^i)} = -\frac{p - 1}{N}
\end{align*}
and \begin{align*}
\sum_{\lambda \in \F_p^*} \sum_{i=0}^{\frac{q - 1}{N} - 1} \zeta_p^{\Tr_{q/p}(\lambda c \theta^i)}\eta(\lambda a \theta^i)  = \frac{p - 1}{N} \eta(ac)\sum_{\lambda \in \F_q^*}\zeta_p^{\Tr_{q/p}(x)}\eta(x). 
\end{align*}
\end{corollary}

Define the linear codes over $\F_p$ :
\begin{linenomath}\begin{equation}\label{eq:cqn}
\mathcal{C}_{Q,N} := \{(\psi_{c_1}, \psi_{c_2}, \ldots, \psi_{c_n})_{1 \leq i \leq n} \mid (c_1, c_2, \ldots, c_n) \in \mathcal{C}_Q \} \subseteq \F_p^{\frac{q - 1}{N} \times n}, 
\end{equation}\end{linenomath}
\begin{linenomath}\begin{equation}\label{eq:cqn1}
\mathcal{C}_{Q,N}' := \{(\psi_{c_1}, \psi_{c_2}, \ldots, \psi_{c_{n'}})_{1 \leq i \leq {n'}} \mid (c_1, c_2, \ldots, c_{n'}) \in \mathcal{C}_Q' \} \subseteq \F_p^{\frac{q - 1}{N} \times n'},
\end{equation}\end{linenomath}
where $n$ and $n'$ are the length of $\mathcal{C}_Q$ and $\mathcal{C}_Q'$.

By Proposition~3 \cite{XOYM2023} we directly obtain the following proposition.

\begin{proposition}
(1) The code $\mathcal{C}_{Q, N}$ is a $[\frac{(q^M - 1)(q - 1)}{N}, m (m_2 + 1)]_{p}$ linear code
with the weight distribution in Tables 5 and 6.

(2)  The code $\mathcal{C}_{Q, N}'$ is a $[\frac{q^M (q - 1)}{N}, m (m_2 + 2)]_{p}$ linear code
with the weight distribution in Tables 7 and 8.
\end{proposition}

\begin{table}
\centering
\caption{Weight distribution of $\mathcal{C}_{Q, N}$ when $r_{Q}$ is even}
\begin{tabular*}{10.5cm}{@{\extracolsep{\fill}}ll}
\hline
\textrm{Weight} $\gamma$ \qquad& \textrm{Frequency} $A_\gamma$   \\
\hline
$0$ & $1$  \\
$\frac{(p - 1) (q - 1)}{p N} q^{M}$& $q (q^{m_2} - 1)$  \\
$\frac{(p - 1) (q - 1)}{p N} q^{M}\left(1 - \epsilon q^{- \frac{r_{Q}}{2}}\right)$ & $q - 1$  \\
\hline
\end{tabular*}
\end{table}

\begin{table}[ht]
\centering
\caption{Weight distribution of $\mathcal{C}_{Q, N}$ when $r_{Q}$ is odd }
\begin{tabular*}{9cm}{@{\extracolsep{\fill}}ll}
\hline
\textrm{Weight} $\gamma$ & \textrm{Frequency} $A_\gamma$   \\
\hline
$0$ & $1$  \\
$\frac{(p - 1) (q - 1)}{p N} q^{M}$ & $q^{m_2 + 1} - 1$\\
\hline
\end{tabular*}
\end{table}
\begin{table}
\centering
\caption{Weight distribution of $\mathcal{C}_{Q, N}'$ when $r_{Q}$ is even}
\begin{tabular*}{10.5cm}{@{\extracolsep{\fill}}ll}
\hline
\textrm{Weight} $\gamma$ \qquad& \textrm{frequency} $A_\gamma$   \\
\hline
0 & 1  \\
$\frac{p-1}{p N}q^{M+1}$ & $q - 1$  \\
$\frac{(q - 1)(p - 1)}{p N} q^{M} $ & $q^2 (q^{m_2} - 1)$  \\
$\frac{(q - 1)(p - 1)}{p N} q^{M} (1 - \epsilon q^{-\frac{r_{Q}}{2}})$ & $q - 1$  \\
$\frac{p - 1}{p N} q^{M}  (q - 1 + \epsilon q^{- \frac{r_{Q}}{2}})$ & $(q - 1)^2$  \\
\hline
\end{tabular*}
\centering
\caption{Weight distribution of $\mathcal{C}_{Q, N}'$ when $r_{Q}$ is odd}
\begin{tabular*}{10.5cm}{@{\extracolsep{\fill}}ll}
\hline
\textrm{Weight} $\gamma$ \qquad& \textrm{frequency} $A_\gamma$   \\
\hline
$0$ & $1$  \\
$\frac{p-1}{p N}q^{M+1}$ & $q - 1$  \\
$\frac{(q - 1)(p - 1)}{p N} q^{M}$ & $q^2 (q^{m_2} - 1) + q - 1$  \\
$\frac{p - 1}{p N} q^{M} \left(q - 1 - \epsilon q^{\frac{1 - r_{Q}}{2}}\right)$ & $\frac{1}{2} (q - 1)^2$  \\
$\frac{p - 1}{p N} q^{M} \left(q - 1 + \epsilon q^{\frac{1 - r_{Q}}{2}}\right)$ & $\frac{1}{2} (q - 1)^2$  \\
\hline
\end{tabular*}
\end{table}

We now give the weight hierarchies of $\mathcal{C}_{Q,N}$. The following theorem can be obtained by a process similar to Theorem~\ref{thm:gwt1}.

\begin{theorem} Assume $Q(x)$ is a quadratic form on $\E/\mathbb{F}_{q}$ and $\mathcal{C}_{Q, N}$ defined by \eqref{eq:cqn}.
For each $ r $ and $ 1\leq r \leq m(m_2 + 1)$, we have the following.\\
  (1) When $r_{Q}$ is even and $\epsilon = 1$,
    \begin{linenomath}\begin{align*}
    d_r(\mathcal{C}_{Q, N}) =
\begin{cases}
\frac{q^{M} (q - 1)}{p^{r} N} \left(p^{r} - 1\right) \left( 1 - q^{- \frac{r_{Q}}{2}}\right), & r \le m, \\
\frac{q^{M} (q - 1)}{p^{r} N} \left(p^{r} -  1  -  q^{- \frac{r_{Q}}{2}} (q - 1)\right), & r > m.
\end{cases} \end{align*}\end{linenomath}
(2) When $r_{Q}$ is even and $\epsilon = - 1$,
\begin{linenomath}\begin{align*}
    d_r(\mathcal{C}_{Q, N}) = \begin{cases}
\frac{q^{M} (q - 1)}{p^{r} N} \left(p^{r} - 1\right), & r \le m m_{2}, \\
\frac{q^{M} (q - 1)}{p^{r} N} \left(p^{r} -  1 +  q^{- \frac{r_{Q}}{2}} (p^{r - m m_{2}} - 1)\right), & r > m m_{2}.
\end{cases}
\end{align*}\end{linenomath}
(3) When $r_{Q}$ is odd,
\begin{linenomath}\begin{align*}
    d_r(\mathcal{C}_{Q, N}) = \frac{q^{M} (q - 1)}{p^{r} N} \left(p^{r} - 1\right).
\end{align*}\end{linenomath}
\end{theorem}
\begin{proof} Similar to  Theorem \ref{thm:gwt1}, we proved that
for $ 1\leq r \leq m (m_2 + 1)$, the $r$-th generalized Hamming weight of $\mathcal{C}_{Q, N}$ satisfies
\begin{linenomath}\begin{align}
 d_{r}(\mathcal{C}_{Q, N}) & = \frac{(q^M - 1) (q - 1)}{N} - \max\Big\{N(V_r): V_r \in [\mathbb{F}_{q} \times \mathbb{F}_{q_{2}}, r]_{p}\Big\},  
       \end{align}\end{linenomath}
where $N(V_r) = \#\{(x, y, i): \Tr_{q/p} \Big( \left(a Q(x) + \Tr_{q_2/q}(b y)\right) \theta^i \Big) = 0, (x, y) \in \mathbb{F}^\star,0 \leq i < \frac{q - 1}{N}, \forall (a, b) \in V_r \}$,
and $[\mathbb{F}_{q} \times \mathbb{F}_{q_{2}}, r]_{p}$ denotes the set of all $r$-dimensional $\F_p$-subspaces of $\mathbb{F}_{q} \times \mathbb{F}_{q_{2}}$. Thus, it suffices to compute $N(V_r)$.

By the orthogonal property of additive characters and Lemma~\ref{lem:main1}, we have
\begin{linenomath}\begin{align}\label{eq:nv}
&p^r \left(N(V_r) + \frac{q - 1}{N}\right)  = \sum\limits_{(x, y) \in \mathbb{F}} \sum\limits_{i = 0}^{\frac{q-1}{N}-1} \sum\limits_{(a, b) \in V_{r}} \zeta_p^{\Tr_{q/p} \big( \left(a Q(x) + \Tr_{q_2/q}(b y)\right) \theta^i \big)}\\\nonumber
& = \sum\limits_{i = 0}^{\frac{q-1}{N}-1}
\sum\limits_{(a, b) \in V_{r}} \sum\limits_{x \in\F_{q_1}} \zeta_p^{\Tr_{q/p} \big( \theta^i a Q(x)\big)} \sum\limits_{y \in \F_{q_2}} \zeta_p^{\Tr_{q_2/p} \big(\theta^i by \big)}  \\ \nonumber
& = \frac{q - 1}{N} q^M + q^{m_2} \sum\limits_{i = 0}^{\frac{q-1}{N}-1} \sum\limits_{\substack{(a, 0) \in V_r\\ a \ne 0}} \sum\limits_{x \in \F_{q_1}} \zeta_p^{\Tr_{q/p} \big( \theta^i a Q(x)\big)}\\
& = \begin{cases}\frac{q - 1}{N} q^M(1 + \epsilon t q^{- \frac{r_{Q}}{2}}), & 2 \mid r_{Q}, \\
\frac{q - 1}{N} q^M + (- 1)^{m - 1} \varepsilon_Q q^{M} (p^*)^{- \frac{m r_{Q}}{2}} \sum\limits_{i = 0}^{\frac{q-1}{N} -1}  \eta(\theta^i) \sum\limits_{\substack{(a, 0) \in V_r\\ a \ne 0}} \eta(a), & 2\nmid r_{Q},
\end{cases}\nonumber
\end{align}\end{linenomath}
where $t = \#\{(a, 0) \in V_r: a \ne 0\}$. Next we discuss case by case according to the parity of $r_Q$.

(1) \ When $r_{Q}$ is even and $\epsilon = 1$, from \eqref{eq:nv} we have \[
N(V_r) = \frac{(q - 1)q^M}{p^rN} \left(1 + t q^{- \frac{ r_{Q}}{2}}\right) - \frac{q - 1}{N}.
\]
 In this case, $N(V_r)$ reaches its maximum value if and only if $t$ is at its maximum.
Obviously, if $1 \leq r \leq m$, 
then $t$ could take its maximum value $p^r -1$ and
$$
d_r(\mathcal{C}_{Q, N}) = \frac{q^{M} (q - 1)}{p^{r} N} \left(p^{r} - 1\right) \left( 1 - q^{- \frac{r_{Q}}{2}}\right).
$$
If $r > m$, 
then $t$ could take its maximum value
$q -1 $ and
$$
d_r(\mathcal{C}_{Q, N}) = \frac{q^{M} (q - 1)}{p^{r} N} \left(p^{r} -  1  -  q^{- \frac{r_{Q}}{2}} (q - 1)\right).
$$

(2) \ When $r_{Q}$ is even and $\epsilon = -1 $, from \eqref{eq:nv} we have \[
N(V_r) = \frac{q - 1}{N} \frac{q^M}{p^r} \left(1 - t q^{- \frac{ r_{Q}}{2}}\right) - \frac{q - 1}{N}.
\]
Hence $N(V_r)$ reaches its maximum value if and only if $t$ is at its minimum.
It's easy to see that if $1 \leq r \leq mm_2$, 
then $t$ could take its minimum value $ 0 $ and 
$$
d_{r}(\mathcal{C}_{Q, N}) = \frac{q^{M} (q - 1)}{p^{r} N} \left(p^{r} - 1\right).
$$
If $r > mm_2$,
then $t$ could take its minimum value $ p^{r-mm_2} - 1 $ and
$$
d_{r}(\mathcal{C}_{Q, N}) = \frac{q^{M} (q - 1)}{p^{r} N} \left(p^{r} -  1 +  q^{- \frac{r_{Q}}{2}} (p^{r - m m_{2}} - 1)\right).
$$

(3) \ When $r_{Q}$ is odd, notice that $N$ is a factor of $p - 1$ prime to $\frac{q - 1}{p - 1}$, hence $N$ and $m$ can not be both even. 
When $N$ is odd, $\theta$ is a non-square element, hence
$\sum\limits_{i = 0}^{\frac{q-1}{N}-1} \eta(\theta^i) = 0$.
When $m$ is odd, any primitive element of $\mathbb{F}_{p}$ is not a square element in $\mathbb{F}_{q}$, hence
$
\sum\limits_{\substack{(a, 0) \in V_r\\ a \ne 0}} \eta(a) = 0$.
Basing on the above discussion and \eqref{eq:nv}, we have 
$$
N(V_r) = \frac{(q - 1) q^M}{N p^r} - \frac{q - 1}{N},
$$ 
and the desired result can be obtained directly.
\end{proof}

Next, we consider the weight hierarchy of $\mathcal{C}_{Q,N}'$. 
\begin{theorem} Assume $Q(x)$ is a quadratic form on $\E/\mathbb{F}_{q}$ with rank $r_{Q}$ and $\mathcal{C}_{Q, N}'$ defined in \eqref{eq:cqn1}. For each $ r $ and $ 1\leq r \leq m(m_2 + 2)$, we
we have the following.

(1) If $r_{Q}$ is even and $\epsilon = 1$, then
\begin{linenomath}\begin{align*}
    &d_{r}(\mathcal{C}_{Q, N}')= \\
    &\begin{cases}
    \frac{q^M (q-1) }{p^r N} (p^{r} - 1)( 1 -  q^{-\frac{r_Q}{2}}), & 1 \leq r \leq m, \\
    \frac{q^M (q-1) }{p^r N} \left( p^r - 1 -  q^{-\frac{r_Q}{2}} (q - 1)\right), & m < r \leq m(m_2+1), \\
    \frac{q^{M} }{p^{r} N}\left(p^{r} (q-1) - (q - p^{r - m (m_{2} + 1)}) (1 + q^{- \frac{r_{Q}}{2}} (q - 1))\right), & r > m (m_2 + 1).
    \end{cases}
\end{align*}\end{linenomath}
 
(2) If $r_{Q}$ is even and $\epsilon = - 1$, then
\begin{linenomath}\begin{align*}
&d_{r}(\mathcal{C}_{Q, N}')=\\
    &\begin{cases}
    \frac{q^M}{p^r N} \left(p^{r} - 1\right)\left(q - 1 - q^{- \frac{r_Q}{2}} \right), & 1 \leq r \leq m, \\
    \frac{q^M (q-1) }{p^r N} \left(p^{r} - 1 - q^{- \frac{r_Q}{2}} \right), & m < r \leq m (m_2 + 1), \\
    \frac{q^M}{p^{r} N} \left(p^{r} ( q - 1 ) - (q - p^{r - m (m_{2} + 1)}) (1 + q^{- \frac{r_{Q}}{2}})\right), & r > m (m_{2} + 1).
    \end{cases}
\end{align*}\end{linenomath}

(3) If $r_{Q}$ is odd, then
\begin{linenomath}\begin{align*}
&d_{r}(\mathcal{C}_{Q, N}')=\\
    &\begin{cases}
    \frac{q^M}{p^r N} \left(p^{r} - 1\right)\left(q - 1 - q^{\frac{1 - r_Q}{2}} \right), & 1 \leq r \leq m, \\
    \frac{q^M (q-1) }{p^r N} \left(p^{r} - 1 - q^{\frac{1 - r_Q}{2}} \right), & m < r \leq m (m_2 + 1), \\
    \frac{q^M}{p^{r} N} \left(p^{r} ( q - 1 ) - (q - p^{r - m (m_{2} + 1)}) (1 + q^{ \frac{1 - r_{Q}}{2}})\right), & r > m (m_{2} + 1).
    \end{cases}
\end{align*}\end{linenomath} 
\end{theorem}

\begin{proof}  By an argument similar to  Theorem \ref{thm:gwt1}, it can be shown that 
for $ 1\leq r \leq m (m_2 + 1)$, the $r$-th generalized Hamming weight of $\mathcal{C}_{Q, N}'$ is given by
\begin{linenomath}\begin{align}
 d_{r}(\mathcal{C}_{Q, N}') & = \frac{q^M  (q - 1)}{N} - \max\Big\{N(V_r'): V_r' \in [\mathbb{F}_{q} \times \mathbb{F}_{q_{2}} \times \mathbb{F}_{q}, r]_{p}\Big\},  
       \end{align}\end{linenomath}
where $N(V_r') = \#\{(x, y, i): \Tr_{q/p} \Big( \left(a Q(x) + \Tr_{q_2/q}(b y) + c\right) \theta^i \Big) = 0, (x, y) \in \mathbb{F}^\star,0 \leq i < \frac{q - 1}{N}, \forall (a, b,c) \in V_r' \}$,
and $[\mathbb{F}_{q} \times \mathbb{F}_{q_{2}} \times \mathbb{F}_{q}, r]_{p}$ denotes the set of all $r$-dimensional $\F_p$ - subspaces of $\mathbb{F}_{q} \times \mathbb{F}_{q_{2}} \times \mathbb{F}_{q}$. Thus, it suffices to compute $N(V_r')$.

By the orthogonal property of additive characters, we have
\begin{linenomath}\begin{align}\label{eq:nv1}
& p^r N(V_r')  = \sum\limits_{(x, y) \in \mathbb{F}} \sum\limits_{i = 0}^{\frac{q - 1}{N} - 1} \sum\limits_{(a, b, c) \in V_{r}'} \zeta_p^{\Tr_{q/p} \big( \left(a Q(x) + \Tr_{q_2/q}(b y) + c \right) \theta^i \big)}\\\nonumber
& = \sum\limits_{i = 0}^{\frac{q - 1}{N} - 1} \sum\limits_{(a, b, c) \in V_{r}'} \zeta_p^{\Tr_{q/p} ( \theta^i c )} \sum\limits_{x \in\F_{q_1}} \zeta_p^{\Tr_{q/p} \big( \theta^i a Q(x) \big) } \sum\limits_{y \in \F_{q_2}} \zeta_p^{\Tr_{q_2/p} (\theta^i by )}  \\ \nonumber
& = q^{m_2} \sum\limits_{i = 0}^{\frac{q - 1}{N} - 1} \sum\limits_{(a, 0, c) \in V_{r}'} \zeta_p^{\Tr_{q/p} ( \theta^i c )} \sum\limits_{x \in\F_{q_1}} \zeta_p^{\Tr_{q/p} \big( \theta^i a Q(x) \big) }   \\ \nonumber
& = q^{m_2} \sum\limits_{i = 0}^{\frac{q - 1}{N} - 1} \left( \sum\limits_{\substack{(a, 0, c) \in V_{r}'\\a \ne 0}} \zeta_p^{\Tr_{q/p} ( \theta^i c )} \sum\limits_{x \in\F_{q_1}} \zeta_p^{\Tr_{q/p} \big( \theta^i a Q(x) \big) } + q^{m_1} \sum\limits_{(0, 0, c) \in V_{r}'} \zeta_p^{\Tr_{q/p} ( \theta^i c )}  \right)\\ \nonumber
& = q^{m_2} \sum\limits_{i = 0}^{\frac{q - 1}{N} - 1} \sum\limits_{\substack{(a, 0, c) \in V_{r}'\\a \ne 0}} \zeta_p^{\Tr_{q/p} ( \theta^i c )} \sum\limits_{x \in\F_{q_1}} \zeta_p^{\Tr_{q/p} \big( \theta^i a Q(x) \big) } + q^{M} \sum\limits_{i = 0}^{\frac{q - 1}{N} - 1} \sum\limits_{(0, 0, c) \in V_{r}'} \zeta_p^{\Tr_{q/p} ( \theta^i c )}.
\end{align}\end{linenomath}
By Corollary \ref{cor:le:s1}, we have
\begin{linenomath}\begin{align*}
   \sum\limits_{i = 0}^{\frac{q - 1}{N} - 1} \sum\limits_{(0, 0, c) \in V_{r}'} \zeta_p^{\Tr_{q/p} ( \theta^i c )} & = \frac{q - 1}{N} + \sum\limits_{i = 0}^{\frac{q - 1}{N} - 1} \sum\limits_{\substack{(0, 0, c) \in V_{r}'\\ c \ne 0}} \zeta_p^{\Tr_{q/p} ( \theta^i c )}\\
   & = \frac{q - 1}{N} + \frac{1}{p - 1}\sum\limits_{\substack{(0, 0, c) \in V_{r}'\\ c \ne 0}} \sum\limits_{\lambda\in\F_p^*}\sum\limits_{i = 0}^{\frac{q - 1}{N} - 1}\zeta_p^{\Tr_{q/p} ( \theta^i \lambda c )}\\
    & = \frac{q - 1}{N} - \frac{ 1}{ N} \sum\limits_{\substack{(0, 0, c) \in V_{r}'\\ c \ne 0}} 1.
\end{align*}\end{linenomath}


 From the above discussion, the equation~\eqref{eq:nv1} becomes
\begin{linenomath}\begin{align}\label{eq:vr0}
& p^r N(V_r')    
 = q^{m_2} \sum\limits_{i = 0}^{\frac{q - 1}{N} - 1} \sum\limits_{\substack{(a, 0, c) \in V_{r}'\\a \ne 0}} \zeta_p^{\Tr_{q/p} ( \theta^i c )} \sum\limits_{x \in\F_{q_1}} \zeta_p^{\Tr_{q/p} \big( \theta^i a Q(x) \big) } + \frac{q - 1}{N} q^{M} - \frac{1}{N} q^{M} t_{3} \\ \nonumber
\end{align}\end{linenomath}
where $t_3 = \#\{(0, 0, c) \in V_r': c \ne 0\}$. 

(1) When $r_Q$ is even and $\epsilon = 1 $, combining with Lemma~\ref{lem:main1} and Corollary \ref{cor:le:s1}, from \eqref{eq:vr0} we have
\begin{linenomath}\begin{align*}\label{eq:vr1}
N(V_r') &= p^{-r}\Big(q^{M - \frac{r_{Q}}{2}}\sum\limits_{i = 0}^{\frac{q - 1}{N} - 1}  \sum\limits_{\substack{(a, 0, c) \in V_{r}'\\a \ne 0}} \zeta_p^{\Tr_{q/p} ( \theta^i c )} + \frac{q^{M}}{N} (q - 1 - t_{3})\Big)\\
& = p^{-r}\left(q^{M - \frac{r_{Q}}{2}} \Big(\frac{q - 1}{N} t_{1} - \frac{1}{ N} t_{2}\Big) + \frac{q - 1}{N} q^{M} - \frac{1}{N} q^M t_{3}\right)\\
&=\frac{q^M}{p^r N} \left( \big((q - 1) t_{1} - t_{2} \big) q^{- \frac{r_{Q}}{2}} + q - 1  -  t_{3} \right),
\end{align*}\end{linenomath}
where $t_1 = \#\{(a, 0, 0) \in V_r': a \ne 0\}$ and $t_2 = \#\{(a, 0, c) \in V_r': ac \ne 0\}$. 
It's easy to see that when $1 \leq r \leq m(m_2+1)$, $N(V_r')$ reaches its maximum if and only if $t_1$ is at its maximum and $t_2, t_3$ are at their minimum $0$. 
If $1 \leq r \leq m$, then $t_1$ could take its maximum $ p^{r} - 1 $ and 
$$
d_{r}(\mathcal{C}_{Q, N}') = \frac{q^M (q-1) }{p^r N} (p^{r} - 1)( 1 -  q^{-\frac{r_Q}{2}}).
$$
If $m < r \leq m(m_2+1)$, 
then $t_1$ could take its maximum $ q - 1 $ and
$$
d_{r}(\mathcal{C}_{Q, N}') = \frac{q^M (q-1) }{p^r N} \left( p^r - 1 -  q^{-\frac{r_Q}{2}} (q - 1)\right).
$$
If $r > m(m_2 + 1)$, in this case,  we first consider the $k$-dimensional subspace $K$ of $\F_q \times \F_q$, where $k > m$.
Any $k$-dimensional subspace of $\F_q \times \F_q$ can be written in the following form 
$$K = \langle (a_1, c_1), \ldots, (a_{k_1},  c_{k_1}), (a_{k_1 + 1}, 0), \ldots, (a_ {k_2}, 0), (0, c_{k_2 + 1}), \ldots, (0,  c_{k}) \rangle ,$$
where
$\{a_1, \ldots, a_ {k_2}\}$ and  $\{c_1, \ldots, c_{k_1},  c_{k_2+1}, \ldots , c_k\}$ are $\mathbb{F}_{p}$-linear independent.
Then 
$$t_1= p^{k_2 - k_1} -1, t_2 = p^k - p^{k_2 - k_1} - p^{k - k_2} + 1, t_3=p^{k - k_2} - 1.$$
The $t_1,t_2$ and $t_3$ here corresponds to the value of the $t_1,t_2$ and $t_3$ mentioned earlier.
Thus
\begin{align*}
    &\big((q - 1) t_{1} - t_{2} \big) q^{- \frac{r_{Q}}{2}} + q - 1  -  t_{3} \\
    &= q^{1 - \frac{r_Q}{2}} p^{k_2 - k_1} - (1-q^{-\frac{r_Q}{2}}) p^{k - k_2} - p^k q^{-\frac{r_Q}{2}}   - q^{1 - \frac{r_Q}{2}}  + q.
\end{align*}
The above expression reaches its maximum when $k_1= 0$ and $k_2=m$.
Therefore, when 
\[V_r' = \mathbb{F}_{q} \times \mathbb{F}_{q_{2}} \times \Big\langle c_1, c_2, \ldots, c_{r-m(m_2+1)} \Big\rangle ,\] 
where $\{ c_1, c_2, \ldots, c_{r - m(m_2 + 1)} \}$ are $\mathbb{F}_{p}$-linear independent, then 
\[t_1 = q-1, t_2 = (q-1)(p^{r - m(m_2 + 1)} - 1), t_3 = p^{r - m(m_2 + 1)} - 1\] 
and $$
d_{r}(\mathcal{C}_{Q, N}') =  \frac{q^{M} }{p^{r} N}\left(p^{r} (q-1) - (q - p^{r - m (m_{2} + 1)}) (1 + q^{- \frac{r_{Q}}{2}} (q - 1))\right).
$$

\rm{(2)} When $r_Q$ is even and $\epsilon = -1$, combining with Lemma~\ref{lem:main1} and Corollary \ref{cor:le:s1}, from \eqref{eq:vr0} we have
\begin{linenomath}\begin{equation*}\label{eq:vr2}
N(V_r') = 
\frac{q^M}{p^r N} \left( \big(t_{2} - (q - 1) t_{1}\big) q^{- \frac{r_{Q}}{2}} + q - 1  -  t_{3} \right).
\end{equation*}\end{linenomath}
It's easy to see that when $1 \leq r \leq m(m_2+1)$, $N(V_r')$ reaches its maximum value if and only if $t_2$ is at its maximum and $t_1, t_3$ are at their minimum $0$. If $1 \leq r \leq m$, then $t_2$ could take its maximum value $ p^r - 1 $
and 
$$
d_{r}(\mathcal{C}_{Q, N}') = \frac{q^M}{p^r N} \left(p^{r} - 1\right)\left(q - 1 - q^{- \frac{r_Q}{2}} \right).
$$
If $m < r \leq m(m_2+1)$, 
then $t_2$ could take its maximum value $ q - 1 $
and
$$
d_{r}(\mathcal{C}_{Q, N}') =  \frac{q^M (q-1) }{p^r N} \left(p^{r} - 1 - q^{- \frac{r_Q}{2}} \right).
$$
If $r > m(m_2+1)$, then 
using the same symbols as in (1), we could obtain
\begin{align*}
    & \big(t_{2} - (q - 1) t_{1}\big) q^{- \frac{r_{Q}}{2}} + q - 1  -  t_{3} \\
    &= - (1 + q^{-\frac{r_Q}{2}} )p^{k - k_2}  - q^{1 - \frac{r_Q}{2}} p^{k_2 - k_1} + p^k q^{-\frac{r_Q}{2}} + q^{1 - \frac{r_Q}{2}}  + q.
\end{align*}
In fact, according to the construction by $K$, $k_1 \leq k_2 \leq m$ and $k_1 + (k - k_2) \leq m,$ that is $k_2 - k_1 \geq k - m.$
Then the above expression reaches its maximum when $k_1 = 2m-k $ and $k_2 = m$.
Therefore,
when 
\begin{multline*}
V'_r = (\{0\} \times \mathbb{F}_{q_2} \times \{0\}) \oplus \\
(\langle (a_1, 0, c_1), (a_2, 0, c_2), \ldots, (a_m, 0, c_m), (0, 0, c_1), \ldots, (0, 0, c_{r-m(m_2+1)}) \rangle),
\end{multline*}
where 
$\{a_1, a_2, \ldots, a_m\}$ and  $\{c_1, c_2, \ldots, c_m\}$ are both $\mathbb{F}_{p}$-linear independent, then
\[t_1 = p^{r-m(m_2+1)} - 1, t_2 = p^{r- mm_2} - 2p^{r-m(m_2+1)} + 1, t_3 = p^{r-m(m_2+1)} - 1\] 
and
$N(V_r')$ could take its maximum and
$$
d_{r}(\mathcal{C}_{Q, N}') =\frac{q^M}{p^{r} N} \left(p^{r} ( q - 1 ) - (q - p^{r - m (m_{2} + 1)}) (1 + q^{- \frac{r_{Q}}{2}})\right).
$$

(3) When $r_Q$ is odd, combining with Lemma~\ref{lem:basic}, Lemma~\ref{lem:main1}  and Corollary \ref{cor:le:s1}, \eqref{eq:vr0} becomes
\begin{linenomath}\begin{align*}\label{eq:vr1}
 &p^{r}N(V_r') =\frac{q^{M}}{N} (q - 1 - t_{3}) \\
 & + q^{M}(-1)^{m-1}\varepsilon_Q(p^*)^{- \frac{m r_{Q}}{2}}\sum\limits_{i = 0}^{\frac{q - 1}{N} - 1}\left(\sum\limits_{\substack{(a, 0, 0) \in V_{r}'\\a \ne 0}} \eta(- \theta^i a) + \sum\limits_{\substack{(a, 0, c) \in V_{r}'\\ac \ne 0}} \zeta_p^{\Tr_{q/p} ( \theta^i c )}\eta(- \theta^i a)\right)  \\
 &= q^{M}(-1)^{m-1}\varepsilon_Q(p^*)^{- \frac{m r_{Q}}{2}}\sum\limits_{\substack{(a, 0, c) \in V_{r}'\\ac \ne 0}}\eta(- a/c)\sum\limits_{i = 0}^{\frac{q - 1}{N} - 1}  \zeta_p^{\Tr_{q/p} ( \theta^i c )}\eta( \theta^i c) + \frac{q^{M}}{N} (q - 1 - t_{3})\\
 &= \frac{q^{M}}{N}\epsilon q^{\frac{1 - r_{Q}}{2}}\sum\limits_{\substack{(a, 0, c) \in V_{r}'\\ac \ne 0}}\eta( ac) + \frac{q^{M}}{N} (q - 1 - t_{3}).
\end{align*}\end{linenomath}
It's easy to see that when $1 \leq r \leq m$, taking 
$$V_r' = \Big\langle (a_1, 0, ka_1), (a_2, 0, ka_2), \cdots, (a_r, 0, ka_r) \Big\rangle, $$
where $a_1, a_2, \ldots, a_r$ are $\mathbb{F}_{p}$-linear independent and $k\in \F_q^*$  satisfying $\eta(k) = \epsilon$, $N(V_r')$ could reach its maximum $\frac{q^{M}}{p^rN}(p^{r} - 1) q^{\frac{1 - r_{Q}}{2}} + \frac{q^{M}}{p^rN} (q - 1)$ and $$d_{r}(\mathcal{C}_{Q, N}') = \frac{q^M}{p^r N} \left(p^{r} - 1\right)\left(q - 1 - q^{\frac{1 - r_Q}{2}} \right).$$
When $m \leq r \leq m(m_2+1)$, taking 
$$V_r' = 
\Big\langle (a_1, 0, ka_1), (a_2, 0, ka_2), \cdots, (a_m, 0, k a_m), (0, b_1, 0), \cdots (0, b_{r-m}, 0) \Big\rangle, $$
where $a_1, a_2, \ldots, a_m$ and $b_1, b_2, \ldots, b_{r-m}$ are both $\mathbb{F}_{p}$-linear independent and $k\in \F_q^*$  satisfying $\eta(k) = \epsilon$, 
$N(V_r')$ could reach its maximum $\frac{q^{M}}{p^rN}(q - 1) q^{\frac{1 - r_{Q}}{2}} + \frac{q^{M}}{p^rN} (q - 1)$ and $$d_{r}(\mathcal{C}_{Q, N}') = \frac{q^M(q - 1)}{p^r N}\left(p^{r} - 1 - q^{\frac{1 - r_Q}{2}} \right).$$
\end{proof}

\section{Conclusion}

Although the weight distributions of linear codes have been extensively studied, knowledge of their generalized Hamming weight (GHW) remains limited. Following the initial work of Jian et al. \cite{JF17} on GHWs via skew sets, several researchers including Li \cite{LF18,LF21}, Liu et al. \cite{LZW23}, Liu and Wang \cite{LW19}, Li and Li \cite{LL21,LL22dcc,LL22dm}, Lu \cite{LWWZ25}, and Li et al. \cite{LCQ22} have investigated GHWs for specific codes constructed from defining sets associated with quadratic functions, cyclotomic classes, or cryptographic functions. However, there are still relatively few results on the generalized Hamming weights of linear codes derived from the first generic construction.

In this paper, we construct two classes of linear codes from quadratic forms via the first generic construction. These codes are shown to be one-weight, two-weight, or four-weight codes. We completely determine their complete weight enumerators and weight hierarchies. Some of the resulting codes are minimal, making them suitable for constructing secret sharing schemes with interesting access structures. Furthermore, we also determine the weight hierarchies of $\mathcal{C}_{Q, N}$ and $\mathcal{C}_{Q, N}'$, which are derived from $\mathcal{C}_{Q}$ and $\mathcal{C}_{Q}'$.


\begin{thebibliography}{12}

\bibitem{AABA1998}
A.~Ashikhmin and A.~Barg.
\newblock Minimal vectors in linear codes.
\newblock {\em IEEE Transactions on Information Theory}, 44(5):2010--2017,
  1998.

\bibitem{B19}
P.~Beelen.
\newblock A note on the generalized hamming weights of reed-muller codes.
\newblock {\em Applicable Algebra In Engineering Communication And Computing},
  30:233--242, 2019.

\bibitem{BLV14}
M.~Bras-Amors, K.~Lee, and A.~Vico-Oton.
\newblock New lower bounds on the generalized hamming weights of ag codes.
\newblock {\em IEEE Transactions on Information Theory}, 60:5930--5937, 2014.

\bibitem{CG1984}
A.~R. Calderbank and J.~M. Goethals.
\newblock Three-weight codes and association schemes.
\newblock {\em Philips Journal of Research}, 39:143--152, 1984.

\bibitem{CK1986}
A.~R. Calderbank and W.~M. Kantor.
\newblock The geometry of two-weight codes.
\newblock {\em Bulletin of the London Mathematical Society}, 18:97--122, 1986.

\bibitem{CDY05}
C.~Carlet, C.~Ding, and J.~Yuan.
\newblock Linear codes from perfect nonlinear mappings and their secret sharing
  schemes.
\newblock {\em IEEE Transactions on Information Theory}, 51(6):2089--2102,
  2005.

\bibitem{ChenMes2023}
R.~Chen and S.~Mesnager.
\newblock Evaluation of weil sums for some polynomials and associated quadratic
  forms.
\newblock {\em Cryptography and Communications}, 15(3):661--673, 2023.

\bibitem{CC97}
J.~Cheng and C.~Chao.
\newblock On generalized hamming weights of binary primitive bch codes with
  minimum distance one less than a power of two.
\newblock {\em IEEE Transactions on Information Theory}, 43(1):294--298, 1997.

\bibitem{CGHFRS85}
B.~Chor, O.~Goldreich, J.~Hastad, J.~Friedmann, S.~Rudish, and R.~Smolesky.
\newblock The bit extraction problem or $t$-resilient functions.
\newblock In {\em 26th Annual Symposium on Foundations of Computer Science
  (sfcs 1985)}, pages 396--407. IEEE, 1985.

\bibitem{DHKW2007}
C.~Ding, T.~Helleseth, T.~Kl{\o}ve, and X.~Wang.
\newblock A generic construction of cartesian authentication codes.
\newblock {\em IEEE Transactions on Information Theory}, 53(6):2229--2235,
  2007.

\bibitem{DLLZ16}
C.~Ding, N.~Li, C.~Li, and Z.~Zhou.
\newblock Three-weight cyclic codes and their weight distributions.
\newblock {\em Discrete Mathematics}, 39(2):415--427, 2016.

\bibitem{FL07}
K.~Feng and J.~Luo.
\newblock Value distributions of exponential sums from perfect nonlinear
  functions and their applications.
\newblock {\em IEEE Transactions on Information Theory}, 53(9):3035--3041,
  2007.

\bibitem{F94}
G.~D. Forney.
\newblock Dimension/length profiles and trellis complexity of linear block
  codes.
\newblock {\em IEEE Transactions on Information Theory}, 40(6):1741--1752,
  1994.

\bibitem{HP98}
P.~Heijnen and R.~Pellikaan.
\newblock Generalized hamming weights of $q$-ary reed-muller codes.
\newblock {\em IEEE Transactions on Information Theory}, 44(1):181--196, 1998.

\bibitem{HK2006}
T.~Helleseth and A.~Kholosha.
\newblock Monomial and quadratic bent functions over the finite field of odd
  characteristic.
\newblock {\em IEEE Transactions on Information Theory}, 52(5):2018--2032,
  2006.

\bibitem{HKM77}
T.~Helleseth, T.~Kl{\o}ve, and J.Mykkeltveit.
\newblock The weight distribution of
 irreducible cyclic codes with block lengths $n_1((q^l-1)/N)$.
\newblock {\em Discrete Mathematics}, 18:179--211, 1977.

\bibitem{HLL24}
S.~Hu, X.~Li, and F.~Li.
\newblock Weight hierarchies of three-weight $p$-ary linear codes from
  inhomogeneous quadratic functions.
\newblock {\em Applicable Algebra In Engineering Communication And Computing} 36:877--896,
  2025.

\bibitem{JL97}
H.~Janwa and A.~K. Lal.
\newblock On the generalized hamming weights of cyclic codes.
\newblock {\em IEEE Transactions on Information Theory}, 43(1):299--308, 1997.

\bibitem{JF17}
G.~Jian, R.~Feng, and H.~Wu.
\newblock Generalized hamming weights of three classes of linear codes.
\newblock {\em Finite Fields And Their Applications}, 45:341--354, 2017.

\bibitem{KTFL93}
T.~Kasami, T.~Takata, T.~Fujiwara, and S.~Lin.
\newblock On the optimum bit orders with respect to the state complexity of
  trellis diagrams for binary linear codes.
\newblock {\em IEEE Transactions on Information Theory}, 39(1):242--245, 1993.

\bibitem{K78}
T.~Kl{\o}ve.
\newblock The weight distribution of linear codes over $\textrm{GF}(q^l)$
  having generator matrix over $\textrm{GF}(q)$.
\newblock {\em Discrete Mathematics}, 23:159--168, 1978.

\bibitem{LLQ09}
C.~Li, S.~Ling, and L.~Qu.
\newblock On the covering structures of two classes of linear codes from
  perfect nonlinear functions.
\newblock {\em IEEE Transactions on Information Theory}, 55(1):70--82, 2009.

\bibitem{LF18}
F.~Li.
\newblock A class of cyclotomic linear codes and their generalized hamming
  weights.
\newblock {\em Applicable Algebra In Engineering Communication And Computing},
  29:501--511, 2018.

\bibitem{LF21}
F.~Li.
\newblock Weight hierarchy of a class of linear codes relating to
  non-degenerate quadratic forms.
\newblock {\em IEEE Transactions on Information Theory}, 67:124--129, 2021.

\bibitem{LL21}
F.~Li and X.~Li.
\newblock Weight distributions and weight hierarchies of two classes of binary
  linear codes.
\newblock {\em Finite Fields And Their Applications}, 73:101865, 2021.

\bibitem{LL22dcc}
F.~Li and X.~Li.
\newblock Weight distributions and weight hierarchies of a family of p-ary
  linear codes.
\newblock {\em Designs Codes and Cryptography}, 90:49--66, 2022.

\bibitem{LL22dm}
F.~Li and X.~Li.
\newblock Weight hierarchies of a family of linear codes associated with
  degenerate quadratic forms.
\newblock {\em Discrete Mathematics}, 345:112718, 2022.

\bibitem{LCL24}
X.~Li, Z.~Chen, and F.~Li.
\newblock Some three-weight linear codes and their complete weight enumerators
  and weight hierarchies.
 \newblock {\em Designs Codes and Cryptography}, to appear. 2026.
\newblock 
  {\href{https://arxiv.org/abs/2410.02209v1}{https://arxiv.org/abs/2410.02209v1}}.
  
 \bibitem{LCQ22} 
  K.~Li, H.~Chen, L.~Qu.
\newblock  Generalized Hamming weights of linear codes from cryptographic functions. 
\newblock {\em Advances in
Mathematics of Communications}, 16(4):859--877, 2022.

\bibitem{LZW23}
C.~Liu, D.~Zheng, and X.~Wang.
\newblock Generalized hamming weights of linear codes from quadratic forms over
  finite fields of even characteristic.
\newblock {\em IEEE Transactions on Information Theory}, 69:5676--5686, 2023.

\bibitem{LW19}
Z.~Liu and J.~Wang.
\newblock Notes on generalized hamming weights of some classes of binary codes.
\newblock {\em Cryptography And Communications}, 10:645--657, 2020.

\bibitem{LWWZ25}
W.~Lu, Q.~Wang, X.~Wang, D.~Zheng. 
\newblock The weight hierarchies of three classes of linear codes. 
\newblock {\em Designs Codes and Cryptography}, 93(5):1337--1355, 2025.

\bibitem{O2018}
O.~Olmez.
\newblock A link between combinatorial designs and three-weight linear codes.
\newblock {\em Designs, Codes and Cryptography}, 86:817--833, 2018.

\bibitem{TXF2017}
C.~Tang, C.~Xiang, and K.~Feng.
\newblock Linear codes with few weights from inhomogeneous quadratic functions.
\newblock {\em Designs, Codes and Cryptography}, 83(3):691--714, 2017.

\bibitem{T2007}
K.~Torleiv.
\newblock {\em Codes for Error Detection}, volume~2.
\newblock World Scientific, 2007.

\bibitem{TV95}
M.~A. Tsfasman and S.~G. Vl$\check{a}$dut.
\newblock Geometric approach to higher weights.
\newblock {\em IEEE Transactions on Information Theory}, 41:1564--1588, 1995.

\bibitem{WZ94}
Z.~Wan.
\newblock The weight hierarchies of the projective codes from nondegenerate
  quadrics.
\newblock {\em Designs Codes and Cryptography}, 4:283--300, 1994.

\bibitem{WW97}
Z.~Wan and X.~Wu.
\newblock The weight hierarchies and generalized weight spectra of projective
  codes from degenerate quadratics.
\newblock {\em Discrete Mathematics}, 177:223--243, 1997.

\bibitem{WJ91}
V.~K. Wei.
\newblock Generalized hamming weights for linear codes.
\newblock {\em IEEE Transactions on Information Theory}, 37(5):1412--1418,
  1991.

\bibitem{XOYM2023}
X.~Xie, Y.~Ouyang, and M.~Mao.
\newblock Vectorial bent functions and linear codes from quadratic forms.
\newblock {\em Cryptography and Communications}, 15(5):1011--1029, 2023.

\bibitem{XLG16}
M.~Xiong, S.~Li, and G.~Ge.
\newblock The weight hierarchy of some reducible cyclic codes.
\newblock {\em IEEE Transactions on Information Theory}, 62:4071--4080, 2016.

\bibitem{YL15}
M.~Yang, J.~Li, K.~Feng, and D.~Lin.
\newblock Generalized hamming weights of irreducible cyclic codes.
\newblock {\em IEEE Transactions on Information Theory}, 61:4905--4913, 2015.

\bibitem{YYZ17}
S.~Yang, Z.~Yao, and C.~Zhao.
\newblock The weight distributions of two classes of $p$-ary cyclic codes with
  few weights.
\newblock {\em Finite Fields and Their Applications}, 44:76--91, 2017.

\bibitem{YCD06}
J.~Yuan, C.~Carlet, and C.~Ding.
\newblock The weight distributions of a classes of linear codes from perfect
  nonlinear functions.
\newblock {\em IEEE Transactions on Information Theory}, 52(2):712--717, 2006.

\bibitem{YD2006}
J.~Yuan and C.~Ding.
\newblock Secret sharing schemes from three classes of linear codes.
\newblock {\em IEEE Transactions on Information Theory}, 52(1):206--212, 2006.


\end{thebibliography}
\end{document}